\documentclass[acmsmall,nonacm]{acmart}

\usepackage[utf8]{inputenc}
\usepackage{graphicx}
\usepackage{xcolor}
\usepackage{amsmath}
\usepackage{amsfonts}
\usepackage{amsthm}
\usepackage{url}
\usepackage{epic,eepic,latexsym,color}
\usepackage{ifthen,graphics,epsfig}
\usepackage{wrapfig}
\usepackage{pdfpages}
\usepackage{shuffle}
\usepackage{xspace}

\usepackage{multirow}
\usepackage{pifont}
\usepackage{xcolor}

\newcommand{\cmark}{\textcolor{green!60!black}{\ding{51}}} 
\newcommand{\xmark}{\textcolor{red}{\ding{55}}}            

\newtheorem{theorem}{Theorem}[section]
\newtheorem{lemma}{Lemma}[section]
\newtheorem{corollary}{Corollary}[section]
\newtheorem{definition}{Definition}[section]

\newtheorem{claim}{Claim}[section]

\newcounter{linecounter}
\newcommand{\linenumbering}{\ifthenelse{\value{linecounter}<10}{(0\arabic{linecounter})}{(\arabic{linecounter})}}
\renewcommand{\line}[1]{\refstepcounter{linecounter}\label{#1}\linenumbering}
\newcommand{\resetline}[1]{\setcounter{linecounter}{0}#1}
\renewcommand{\thelinecounter}{\ifnum \value{linecounter} >
9\else 0\fi \arabic{linecounter}}

\newcommand {\yes} {{\sf YES}\xspace}
\newcommand {\no} {{\sf NO}\xspace}
\newcommand {\maybe} {{\sf MAYBE}\xspace}

\newcommand {\true} {{\sf true}\xspace}
\newcommand {\false} {{\sf false}\xspace}

\newcommand {\snap} {{\sf Snapshot}\xspace}

\newcommand {\adv} {{\mathbb A}\xspace}
\newcommand {\advt} {{\mathbb A}^\tau\xspace}

\newcommand {\sd} {{\sf SD}\xspace}
\newcommand {\wad} {{\sf WAD}\xspace}
\newcommand {\wod} {{\sf WOD}\xspace}
\newcommand {\wde} {{\sf WD}\xspace}
\newcommand {\psd} {{\sf PSD}\xspace}
\newcommand {\pwd} {{\sf PWD}\xspace}

\begin{document}

\title{Asynchronous Fault-Tolerant Language Decidability for
Runtime Verification of Distributed Systems}

\author{Armando Casta\~neda}
\affiliation{
  \institution{Instituto de Matem\'aticas, Universdad Nacional Aut\'onoma de M\'exico}
  \city{Mexico City}
  \country{M\'exico}
}
\email{armando.castaneda@im.unam.mx}

\author{Gilde Valeria Rodr\'iguez}
\affiliation{
  \institution{Posgrado en Ciencia e Ingenier\'ia de la Computaci\'on,
      Universidad Nacional Aut\'onoma de M\'exico}
  \city{Mexico City}
  \country{M\'exico}
  }
\email{gildevroji@gmail.com}

\renewcommand{\shortauthors}{Casta\~neda and Rodr\'iguez}

\begin{abstract}
 Implementing correct distributed systems is an error-prone task.
Runtime Verification (RV) offers a lightweight formal method 
to improve reliability by monitoring system executions against correctness properties.
However, applying RV in distributed settings—where no 
process has global knowledge—poses
fundamental challenges, particularly under full asynchrony and fault tolerance.
This paper addresses the Distributed Runtime Verification (DRV) problem under such conditions.
In our model, each 
process in a distributed monitor
receives a fragment 
of the input \emph{word} describing 
system behavior and must decide whether this word belongs to the 
\emph{language} representing the correctness property being verified. 
Hence, the goal is to decide languages in a distributed fault-tolerant manner. 
We propose several decidability definitions, study the relations among them,
and prove possibility and impossibility results. One of our main results is a 
characterization of the correctness properties that can be decided asynchronously. 
Remarkably, it applies to \emph{any} language decidability definition.
Intuitively, the characterization is that only properties with \emph{no real-time order constraints} can be
decided in asynchronous fault-tolerant settings.
These results expose the expressive limits of DRV in realistic systems, 
as several properties of practical interest 
rely on reasoning about real-time order of events in executions.
To overcome these limitations, we introduce a weaker model where the
system under inspection is verified \emph{indirectly}.
Under this weaker model
we define \emph{predictive decidability}, 
a decidability definition that turn some real-time sensitive correctness properties verifiable.
Our framework unifies and 
extends existing DRV theory and sharpens the boundary of 
runtime monitorability under different assumptions.
\end{abstract}

\maketitle

\section{Introduction}

Implementing correct distributed systems is an error-prone task. 
Faults often stem from subtle interactions between components and
 the need to reason about exponentially many possible executions,
  due to communication delays or process failures. Even after extensive testing, 
  flaws can persist in production systems.

Runtime verification~\cite{BF18, FHR13, HG05, LS09} addresses this
 challenge by dynamically checking system correctness through 
 \emph{monitors}—algorithms that determine whether the current 
 execution satisfies a correctness criterion. Monitors can be 
 applied to centralized or distributed systems, in both hardware and software.

In \emph{distributed runtime verification}, the monitors themselves 
are distributed systems. The central challenge~\cite{BFRT16} 
is that processes have only \emph{partial views} of the execution 
but must reach \emph{consistent global decisions}—a problem 
exacerbated in \emph{asynchronous fault-tolerant} environments
due to the impossibility of consensus~\cite{FLP85}. This work 
focuses precisely on that setting: fault-tolerant, asynchronous 
distributed runtime verification of distributed systems.

Prior research has focused on synchronous or semi-synchronous 
settings (e.g.,\cite{AFIMP20,BGKS20,GXJLSBH22,RF22}) or 
asynchronous systems without faults (e.g.,\cite{CGNM13,FFY08,SS14,ETQ05}), 
using either message-passing or shared-memory models. A few asynchronous, 
fault-tolerant shared-memory monitors exist under strong timing 
assumptions about the system being verified~\cite{BFRR22,FRT13,FRT20}, 
hence these are \emph{not fully} asynchronous (see Section~\ref{sec:related-work}).

However, a key issue has been largely overlooked: many 
correctness properties—such as \emph{linearizability}~\cite{HW90},
 \emph{sequential consistency}~\cite{L79}, and some
  \emph{eventual consistency} definition~\cite{ShapiroPBZ11,Vogels09}—require 
  reasoning about \emph{real-time order} of events,
  which in general is impossible to detect under full asynchrony~\cite{L78}.
  Recent work~\cite{podc23,rv24} 
  introduced the first fully asynchronous,  fault-tolerant monitors for linearizability,
  that address this real-time order problem for the first time.

This paper generalizes those insights. We study the distributed
 runtime verification of distributed systems under full asynchrony
  and crash faults, explicitly addressing the problem of distinguishing 
  real-time order. We extend prior work~\cite{podc23,rv24,BFRR22,FRT20}
   to broader correctness criteria and illustrate our results on objects
    such as ledgers, counters, and registers.
 
We propose a unified framework for distributed runtime 
verification in asynchronous,  shared-memory, fault-tolerant settings 
(Sections~\ref{sec:dist-lang} to~\ref{sec:asynch-decidability}). 
The central insight is to formalize the system under verification as an asynchronous
 adversary ($\adv$), which captures nondeterministic delays and crashes. 

In our model, each process receives a subsequence with invocations and responses of the input \emph{word} describing 
the system behavior, and must decide whether this word belongs to the 
\emph{language} representing the correctness property that is verified. Hence, the goal
 is to \emph{decide} a language in a distributed, wait-free manner. We 
 formalize this through a notion of decidability, where each execution 
 leads processes to emit a verdict (e.g., $\yes$, $\no$, $\maybe$).
 We  primarily focus on decidability with two-valued verdicts,  $\yes$ and $\no$,
 and  define \emph{strong decidability} as a property
  where all processes eventually and correctly accept or reject the 
  input—capturing classical soundness and completeness requirements in runtime verification~\cite{E-HF18, MB15, NFBB17}.
 We also define \emph{weak decidability}, where verdicts 
 need not be unanimous or immediate.

 We first prove that linearizability, sequential consistency and \emph{strong eventual counters}~\cite{AlmeidaB19}
  are in general neither strongly decidable nor weakly decidable 
  (Lemma~\ref{lemma:no-wd-lin-sc}, 
  Corollaries~\ref{coro:no-sd-lin-sc}, \ref{coro:no-sd-lin-sc-2} and~\ref{coro:no-wd-lin-sc}).
   We also show that
\emph{weak eventual counters}~\cite{AlmeidaB19}
 are not strongly decidable 
 (Lemma~\ref{lemma:no-sd-evc}) 
 but are weakly decidable (Lemma~\ref{lemma:wd-evc}).
 This establishes that strong decidability is strictly contained in weak decidability (Theorem~\ref{theo:separation-sd-wd}).

We then present one of our main results: a characterization
stating that only properties without real-time constraints can be decided under $\adv$ (Theorem~\ref{theo:rt-oblivious}). 
Remarkably, this result applies to \emph{any} decidability definition.
The characterization thus proves \emph{inherent limitations} for asynchronous, fault-tolerant runtime verification:  
it is simply impossibly to verify real-time sensitive correctness properties such as linearizability, 
regardless of the decidability definition. 

To circumvent this limitation, we propose a relaxed model where the 
adversary $\adv$ is replaced by a \emph{timed} adversary ($\advt$), which
 is obtained from $\adv$ by attaching \emph{timestamps} to  $\adv$'s events
(Section~\ref{sec:timed-adversaries}). 
The weaker adversary $\advt$ enables new forms of verification.
Prior work~\cite{podc23,rv24} already explored this \emph{indirect verification} regime
to weakly verify linearizability in asynchronous, fault-tolerant settings,
under a decidability definition that we call here \emph{predictive strong decidability}.
The key idea is to leverage 
 \emph{indistinguishability}: even if $\adv$ cannot be verified,
  the timestamps from $\advt$ allow processes to reason about time intervals 
  in which operations of the verified system occur in an execution.
Predictive strong decidability allows false negatives—incorrect 
executions however never go undetected—but ensures that any false negative verdict
 is justified by a valid timestamp-based explanation  that justifies the existence
 of an execution of $\advt$ that deviates from the correctness property that is verified.

Unlike~\cite{podc23,rv24}, we explore more in detail what can be verified under $\advt$.
We extend this framework by showing that, while in general linearizability 
is predictively strongly decidable, already shown in~\cite{podc23,rv24}, strong eventual counters 
are not (Lemma~\ref{lemma:no-psd-evc}). Then, we introduce
\emph{predictive weak decidability}, a weaker notion, and prove that strong
 eventual counters fall into this relaxed class 
 (Lemma~\ref{lemma:couner-in-pwd}). This implies 
 a strict hierarchy: predictive strong decidability
  is strictly contained in predictive weak decidability
   (Theorem~\ref{theo:psd-in-pwd}). Still, we show that
    some objects—such as the \emph{eventual ledger}~\cite{ledger}—remain 
    undecidable even under this relaxed decidability notion (Lemma~\ref{lemma:counter-no-pwd}).
    Table~\ref{tab:detailed-results} summarizes the results discussed or proved in the paper, except the characterization.

  \begin{table}[hb]
    \caption{Results summary.}
    \label{tab:detailed-results}
    \centering
    \small
    \begin{tabular}{l|cc|cc}
      \multirow{2}{*}{\textbf{Language / Property}} & 
      \multicolumn{2}{c|}{\textbf{($\adv$)}} & 
      \multicolumn{2}{c}{\textbf{($\advt$)}} \\
       & \textbf{SD} & \textbf{WD} & \textbf{PSD} & \textbf{PWD} \\
      \hline
      Linearizable register (LIN\_REG)           & \xmark & \xmark & \cmark~\cite{podc23} & \cmark \\
      Seq. consistent register (SC\_REG)        & \xmark & \xmark & \cmark~\cite{podc23} & \cmark \\
      Linearizable ledger (LIN\_LED)             & \xmark & \xmark & \cmark~\cite{podc23} & \cmark \\
      Seq. consistent ledger (SC\_LED)           & \xmark & \xmark & \cmark~\cite{podc23} & \cmark \\
      Event. consistent ledger (EC\_LED)           & \xmark & \xmark & \xmark & \xmark \\
      Weak event. counter (WEC\_COUNT)               & \xmark & \cmark & \xmark & \cmark \\
      Strong event. counter (SEC\_COUNT)             & \xmark & \xmark & \xmark & \cmark \\
    \end{tabular}
  \end{table}

Our possibility results use only read/write registers, 
hence can be simulated in asynchronous message-passing systems 
tolerating crash faults in less than half the processes~\cite{ABD95}.
 Our impossibility results hold under operations with arbitrarily
  high \emph{consensus number}~\cite{H91}.

\subsection{Related Work}
\label{sec:related-work}

The notions of soundness (no false positives) and completeness
 (no false negatives) in runtime verification originate from Pnueli and Zaks~\cite{PZ06},
 who formalized monitorability by requiring that, for every finite
 execution prefix, there exists a continuation that determines the 
 property’s satisfaction or violation (i.e., positively or negatively determines the
 infinitary property). 
  Their work inspired generalized definitions of monitorability, such as
   Bauer et al.~\cite{BF16}, who use good/bad prefixes and four-valued truth domains with four 
   variable veredicts. These adaptations extend to 
distributed systems, where decidability hinges on consistent 
   verdicts, ranging from binary (correct/incorrect) to richer $k$-valued classifications~\cite{BFRT16,BFRR22}.
However, prior work often overlooked the semantic gap between system behavior 
and detection capability—particularly in asynchronous settings. 
For instance,~\cite{E-HF18} observed inconsistent verdicts 
(e.g., flipping between ‘correct’ and ‘incorrect’) when verifying
 the same Java program across different runtime monitors. 
 This occurs because these tools ignore real-time operation order, a critical factor in distributed systems.
 Our work addresses this by explicitly modeling asynchrony’s impact on detection.

 Fraigniaud, Rajsbaum, and Travers~\cite{FRT13} pioneered the study
 of asynchronous, fault-tolerant, distributed runtime verification.
 They introduced a \emph{static} model in which the
   distributed system under inspection is in a quiescent state. 
   Each process in the distributed monitor receives a sample of 
   this system state and, after communication via read/write shared 
   memory (excluding stronger primitives), outputs a binary 
   decision—either ${\sf true}$ or ${\sf false}$. Their model 
   has two main limitations, as mostly oriented to prove impossibility results: 
   (1) it is not fully asynchronous,
    as the verified system does not evolve during verification, 
    and (2) it lacks the ability to capture real-time order of
     events, which is essential for properties like linearizability.
      In this model, properties are defined as sets of sets of samples, 
      and a property is decidable when all processes decide ${\sf true}$ 
      if and only if the input set belongs to the property. Our notion of 
      strong decidability extends this to our setting.
      This model was extended in~\cite{FRT20}, which generalized the verification framework in two key ways: 
      outputs were expanded from binary decisions to multiple \emph{opinions} (possible verdicts), and 
      decidability was abstracted to require only that monitors produce distinct 
          verdicts for executions with different property memberships.
      This aligns with our notion of $\mathsf{P}$-decidability in Theorem~\ref{theo:rt-oblivious}. 
       Crucially, they introduced the \emph{alternation number} $k$ of a property, 
        proving that any property with such alternation number can be verified with at most
       $k+1$ opinions, a result we contextualize in Section~\ref{sec:characterization}.
 
 Bonakdarpour, Fraigniaud, Rajsbaum, Rosenblueth, and Travers~\cite{BFRR22}
  extended this work to a \emph{dynamic} setting, albeit with a strong synchrony 
  assumption: the system does not progress until the monitor completes 
  its verification. Again, the model is not fully asynchronous. They 
  generalized the alternation number and proved that any property with
   alternation number $k$ can be verified using at most $2k+4$ opinions.
 
 Fully asynchronous and fault-tolerant monitors were introduced by
  Rodríguez~\cite{thesis-valeria} and Castañeda and Rodríguez~\cite{podc23, rv24}, 
  targeting runtime verification of linearizability in shared-memory object 
  implementations. In their asynchronous model, monitor processes interact directly with 
  the implementation under inspection, invoking operations and receiving responses, 
  with the ability to report $\no$ as soon as incorrect behavior is detected. 
  Crucially, their model supports verification while the system continues 
  evolving—i.e., without requiring quiescence.
 
 They showed that for some objects (e.g., queues and stacks), no 
 monitor can be both sound and complete—i.e., one that outputs 
 $\no$ if and only if the execution is incorrect. This impossibility
  result is captured in our framework by the concept of strong decidability.
   To cope with the impossibility, Rodríguez~\cite{thesis-valeria} proposed 
   a relaxation of completeness, enabling partial verification of linearizability.
    However, the relaxation permits false positives, making it unreliable for 
    detecting incorrect behaviors, which is central to runtime verification.
 
 Castañeda and Rodríguez~\cite{podc23} introduced a complementary relaxation
  of soundness. They showed how to transform any implementation into a related 
  one such that either both are linearizable or neither is. Under this 
  transformation, monitors can be allowed to produce false negatives, 
  provided they eventually exhibit a non-linearizable execution as proof.
   This indirect verification approach enables runtime verification of 
   linearizability for any object, albeit in a relaxed form. They also 
   demonstrated that the original, non-relaxed problem cannot be solved,
    even under this transformation. Their verification algorithms were
     further optimized in~\cite{rv24}, improving step complexity and
      moving closer to practical deployability. Nonetheless, the scope 
      of properties addressed remains limited to linearizability and some of its variants.

\section{Distributed Languages}
\label{sec:dist-lang}

A \emph{distributed alphabet} $\Sigma$ is the union of $n \geq 2$ disjoint \emph{local alphabets}, 
$\Sigma_1, \hdots, \Sigma_n$, with each local alphabet $\Sigma_i$ being the union of two disjoint possibly-infinite alphabets, 
$\Sigma^{<}_i$, the \emph{invocation} alphabet, and $\Sigma^{>}_i$, 
the \emph{response} alphabet.\footnote{Possibly-infiniteness assumption is just by conveniency. }

A \emph{word} over $\Sigma$ is a sequence of symbols in $\Sigma$.
We say that $x$ is a \emph{$\omega$-word} if it has infinitely many symbols. 
The \emph{length} of $x$, denoted $|x|$, is the number of symbols in $x$.
The \emph{local word of $\Sigma_i$} in $x$, denoted $x|i$, is the projection of $x$ over the local alphabet $\Sigma_i$.

\begin{definition}[Well-formed $\omega$-words]
A $\omega$-word $x$ is \emph{well-formed} if for every $x|i$:
\begin{enumerate}
	\item Reliability: $x|i$ is a $\omega$-word.
	\item Sequentiality: $x|i$ alternates symbols in $\Sigma^{<}_i$ and $\Sigma^{>}_i$, 
	starting with $\Sigma^{<}_i$.
	\item Fairness: for every $k \geq 1$, there is a finite prefix of $x$ containing the first $k$ symbols of $x|i$.
\end{enumerate}
The set with all well-formed $\omega$-words over $\Sigma$ is denoted $\Sigma^\omega$.
\end{definition}

\begin{definition}[Distributed languages]
A \emph{distributed language} $L$ over a distributed alphabet~$\Sigma$ is a subset of $\Sigma^\omega$.
\end{definition}

A word of a language models a \emph{concurrent history} where invocations to and responses from 
a \emph{distributed service} (e.g. a shared-memory or message-passing implementation of an object) are interleaved. 
Given a $\omega$-word $x \in \Sigma^\omega$, in every local word $x_i$,
each invocation symbols $v \in \Sigma^{<}_i$ is immediately succeeded 
by a response symbols $w \in \Sigma^{>}_i$ that \emph{matches} $v$.
We call such pair $(v,w)$ an \emph{operation of} $p_i$ in $x$. 
An operation $op$ \emph{precedes} an operation $op'$ in $x$, denoted $op \prec_x op'$, 
if and only if  the response symbols of $op$ appears before the invocation symbol of $op$'. 
The operations are \emph{concurrent}, denoted $op ||_x op'$, if neither $op \prec_x op'$ nor $op' \prec_x op$ hold.
For a finite prefix $x'$ of $x$, an operation is \emph{complete in} $x'$ if and only if both
its invocation and response symbols appear in $x'$, and otherwise it is \emph{pending} in $x'$.

 \paragraph{Example 1.}
Let us consider a \emph{register}, one of the simplest sequential objects, 
which provides two operations: $write(x)$ that writes $x$ in the register
and $read()$ that returns the current value of the register.
The initial state of the register is~0.

We are interested in the linearizable and sequentially consistent concurrent histories of the register.
A finite concurrent history $H$ is \emph{sequentially consistent}~\cite{L79} if and only if responses to pending operation
can be appended to $H$, and the rest of pending operations removed, so that the operations 
of the resulting history $H'$ can be ordered
in a sequential history $S$ that respects process-order and is valid for the register. 
The history $H$ is \emph{linearizable}~\cite{HW90}
if additionally $S$ preserves real-time, namely, if an operation $op$ completes before another operation
$op'$ in $H'$, that order is preserved in $S$.

We model such concurrent histories as a distributed language as follows.
 For each process $p_i$, the local alphabets are $\Sigma^<_i = \{<_i, <^0_i, <^1_i, <^2_i, \hdots \}$ and
$\Sigma^>_i = \{>_i, >^0_i, >^1_i >^2_i, \hdots \}$.
The symbols in $\Sigma^<_i$ and $\Sigma^>_i$ are identified with invocation and responses of $p_i$ as follows:

\begin{itemize}
\item $<^x_i$ is identified with invocation to $write(x)$ of $p_i$;
\item $>_i$ is identified with the response to $write(x)$ of $p_i$ (returning nothing).
\item $<_i$ is identified with invocation to $read()$ of $p_i$;
\item $>^x_i$ is identified with response to $read()$ of $p_i$, returning $x$.
 \end{itemize}
 
 Given this identification, we consider linearizability and sequential consistency of finite words over $\Sigma_\omega$
 to define the corresponding distributed languages for the register. 
  
\begin{definition}[Sequential consistent register]
The language $SC\_REG$ contains every word of $\Sigma^\omega$
such that every finite prefix of it is sequentially consistent with respect to the sequential register.
\end{definition}
  
\begin{definition}[Linearizable register]
The language $LIN\_REG$ contains every word of $\Sigma^\omega$
such that every finite prefix of it is linearizable with respect to the sequential register.
\end{definition}

\paragraph{Example 2.}

In our second example, we consider the \emph{ledger} object in~\cite{ledger},
which is a formalization of the ledger functionality in blockchain systems.
It is an object whose state is a list of items $S$, initially empty, and provides two operations,
$append(r)$ that appends $r \in U$ to $S$, where $U$ is the possibly-infinite universe of \emph{records}
that can be appended, and $get()$ that returns $S$.

We model the concurrent histories of a ledger object as follows.
The invocation and response alphabets of $p_i$ are 
$\Sigma^<_i = \{<_i\} \cup \{<^r_i | r \in U\}$ and
$\Sigma^>_i = \{>_i\} \cup \{ >^s_i | \hbox{$s$ is a finite word over $U$}\}$. 
The symbols in $\Sigma^<_i$ and $\Sigma^>_i$ are identified with invocation and responses of $p_i$ as follows:

\begin{itemize}
\item $<^r_i$ is identified with invocation to $append(r)$ of $p_i$;
\item $>_i$ is identified with the response to $append(r)$ of $p_i$ (returning nothing).
\item $<_i$ is identified with invocation to $get()$ of $p_i$;
\item $>^s_i$ is identified with response to $get()$ of $p_i$, returning string $s$.
 \end{itemize}

\begin{definition}[Sequential consistent ledger]
The language $SC\_LED$ contains every word of $\Sigma^\omega$
such that every finite prefix of it is sequentially consistent with respect to the sequential ledger.
\end{definition}
  
\begin{definition}[Linearizable ledger]
The language $LIN\_LED$ contains every word of $\Sigma^\omega$
such that every finite prefix of it is linearizable with respect to the sequential ledger.
\end{definition}

\paragraph{Example 3.}
We now consider the case of the counter, a sequential object that provides two operations:
$inc()$ that increments by one the current value of the counter, and $read()$ that returns the current value.
The initial state of the counter is~$0$.

We are interested in concurrent histories of the counter that provide only \emph{eventual} guarantees.
There have been proposed different definitions of what an eventual counter is (see for example~\cite{AlmeidaB19,ShapiroPBZ11, Vogels09}). 
Here we consider the following two motivated by~\cite{AlmeidaB19}.

An infinite concurrent history $H$ of a counter is \emph{weakly-eventual consistent} if:
\begin{enumerate}
\item every $read$ operation $op$ of a process returns a value that is at least the
number of $inc$ operations of the same process that precede $op$,

\item every $read$ operation of a process returns a value that is at least
the value returned by the immediate previous $read$ operation of the same process, and

\item for every finite prefix $\alpha$ 
such that the (infinite) suffix $\beta$ has only $read$ operations, eventually all operations in $\beta$
return the number of $inc$ operations in $\alpha$.
\end{enumerate}

For each process $p_i$, the local alphabets are $\Sigma^<_i = \{<_i, <^+_i\}$ and
$\Sigma^>_i = \{>_i, >^0_i >^1_i, \hdots \}$.
The symbols in $\Sigma^<_i$ and $\Sigma^>_i$ are identified with invocation and responses of $p_i$ as follows:

\begin{itemize}
\item $<^+_i$ is identified with invocation to $inc()$ of $p_i$;
\item $>_i$ is identified with the response to $inc()$ of $p_i$ (returning nothing);
\item $<_i$ is identified with invocation to $read()$ of $p_i$;
\item $>^x_i$ is identified with response to $read()$ of $p_i$, returning $x$.
 \end{itemize}

Given this identification, we consider the language that corresponds to the weakly-eventual consistent counter:

\begin{definition}[Weakly-eventual consistent counter]
The language $WEC\_COUNT$ contains every word of $\Sigma^\omega$ that is weakly-eventual consistent 
with respect to the counter.
\end{definition}

An infinite concurrent history $H$ of a counter is \emph{strongly-eventual consistent} if
it satisfies the three properties of the weakly-eventual counter and:

\begin{enumerate}
\item[(4)] every $read$ operation of a process returns a value that is at most to the number
of $inc$ operations that precede or are concurrent to the operation.
\end{enumerate}

Observe that the fourth property is related to the real-time order of operations in $H$.

\begin{definition}[Strongly-eventual consistent counter]
The language $SEC\_COUNT$ contains every word of $\Sigma^\omega$ that is strongly-eventual consistent 
with respect to the counter.
\end{definition}

\paragraph{Example 4.}

In our last example, we consider an eventual consistent ledger object~\cite{ledger}.
An infinite concurrent history $H$ of the ledger objects is \emph{eventually consistent} if for each 
finite prefix $\alpha$ of it: 
\begin{enumerate}
\item it is possible to append response symbols to $\alpha$ to make all operations complete
so that  there is a permutation of the operations  giving a sequential history that is valid for the ledger, and
\item eventually, every $get$ operation in $H$ returns a string that contains the input record of every $append$ in $\alpha$.
\end{enumerate}

Using the identification above, we define the language of the eventual consistent ledger

\begin{definition}[Eventual consistent ledger]
The language $EC\_LED$ contains every word of $\Sigma^\omega$ such that each of its finite prefixes is 
eventually consistent.
\end{definition}

\section{The Computation Model}

We consider a standard concurrent asynchronous system (e.g. \cite{H91,R13}) with $n$ crash-prone processes, 
$p_1, p_2, \hdots, p_n$, each being a state machine, 
possibly with infinitely many states. 
It is assumed that at most $n-1$ processes crashes in an execution of the system.
The processes communicate each other by applying \emph{atomic} operations on a shared memory,
such as simple $read$ and $write$, or more complex and powerful operations such as $test\&set$ or $compare\&swap$.

A \emph{local algorithm} $V_i$ for a process $p_i$ specifies the local or shared memory operations  
$p_i$ executes as a result of its current local state.
A \emph{distributed algorithm} $V$ is a collection of local algorithms, one for each process. 
An operation performed by a process is called \emph{step}. For a step $e$, $p(e)$ denotes the process that performs~$e$.
A \emph{configuration} $C$ is a collection $(s_1, \hdots, s_n, sm)$, where $s_i$ is a state of $p_i$ and $sm$ is a state of the shared memory.
An \emph{initial} configuration has initial processes and shared memory states. 
An \emph{execution} $E$ of $V$ is an infinite sequence $C_0, e_1, C_1, e_2, \hdots$,
where $C_0$ is an initial configuration, and, for every $k \geq 0$, $e_k$ is the step specified by $V_{p(e_k)}$
when $p(e_k)$ is in the state specified in $C_k$, and configuration $C_{k+1}$ reflects the new state of $p(e_k)$ and~the shared memory.
Since the system is asynchronous, 
there is no bound on the number of steps of other processes between consecutive
steps of the same process.

We are interested in distributed algorithms that interact with a distributed service $\adv$
(e.g., a concurrent shared-memory implementation)
in order to runtime verify it. Namely, in every interaction, $\adv$ exhibits one of its possible behaviors,
and the aim is to  determine if the behavior is correct, with respect to a given correctness condition.
Correctness conditions are defined as distributed languages, hence
the ultimate goal of an algorithm is to determine if the current history of $\adv$ belongs to the distributed language,
namely, deciding the language in a distributed manner.

In our model, motivated to the model in~\cite{podc23}, inputs to algorithms
are obtained through an \emph{infinite interaction} between the processes and $\adv$.
Intuitively, each process 
$p_i$ \emph{sends} an invocation symbols in $\Sigma^<_i$ to $\adv$ and, at a later time, $\adv$
\emph{replies} to $p_i$ a response symbol in $\Sigma^>_i$, and this loop repeats infinitely often.
We conceive $\adv$ as a powerful \emph{adversary} that determines the 
invocation symbols processes send to it, the responses it sends to the processes,
and the times when all events happen in an execution.

\begin{figure}[ht]
\centering{ \fbox{
\begin{minipage}[t]{150mm}
\scriptsize
\renewcommand{\baselinestretch}{2.5} \resetline
\begin{tabbing}
aaaa\=aaa\=aaa\=aaa\=aaa\=aaa\=aaa\=\kill 

{\bf Shared Variables:}\\

$~~$  Shared memory $M$\\ \\

{\bf Local algorithm $V_i$ for process $p_i$:}\\

 \> {\bf while} {\sf true} {\bf do}\\

\line{Gen1} \>\> Non-deterministically pick an invocation symbol $v_i \in \Sigma^<_i$\\

\line{Gen2} \>\> Exchange information using $M$\\

\line{Gen3} \>\> Send $v_i$ to the adversary $\adv$\\

\line{Gen4} \>\> Receive response symbol $w_i \in \Sigma^>_i$ from the adversary $\adv$\\

\line{Gen5} \>\> Exchange information using $M$\\

\line{Gen6} \>\> Report a value (e.g. $\yes$, $\no$ or $\maybe$)

\end{tabbing}
\end{minipage}
  }
\caption{Generic structure of algorithms interacting with the adversary $\adv$.}
\label{fig-generic-algo}
}
\end{figure}

More specifically, in an algorithm $V$ interacting with $\adv$, 
each process runs a local algorithm following the generic structure that appears in Figure~\ref{fig-generic-algo}.
In Lines~\ref{Gen2},~\ref{Gen3},~\ref{Gen5} and~\ref{Gen6}, $p_i$ executes \emph{wait-free}~\cite{H91} blocks of code, 
namely, crash-tolerant codes where $p_i$ cannot block because of delays or failures  of other processes.
 It is assumed that in every iteration $p_i$ reports one value
 in the block in Line~\ref{Gen6}.

In the decidability notions we propose in the next section,
we will focus on the fair failure-free executions of $V$.
A failure-free execution $E$ is \emph{fair} if for each process $p_i$  and each integer $k \geq 1$, 
there is a finite prefix of $E$ containing $k$ steps of $p_i$.
 
The \emph{input} to algorithm $V$ in an execution 
$E$ is the subsequence of $E$ with the invocations sent to and responses from $\adv$.
Thus, the input is a $\omega$-word
that is determined by the ``times'' when Lines~\ref{Gen3} and~\ref{Gen4} of processes occur in~$E$, 
which are \emph{local} to the processes and decided by the adversary~$\adv$. 
If $E$ is fair and failure-free, then $x(E) \in \Sigma^\omega$, 
namely, it is a well-formed $\omega$-word.
Since $\adv$ is asynchronous,
for the processes it is impossible to predict when exactly their invocations and responses occur,
hence one of their main target is to communicate each other, in Lines~\ref{Gen2},~\ref{Gen3} and~\ref{Gen6},
 in order to ``figure out'' the input~$x(E)$.
 
 It is assumed that $\adv$ is a \emph{black-box}, namely, there is no information about the internal functionality of $\adv$
 that $V$ could try to exploit in order to runtime verify it. Therefore, we assume that $\adv$ \emph{can exhibit any possible behavior}. 
 More specifically, the set of all possible concurrent histories of $\adv$ is~$\Sigma^\omega$.
 This assumption implies:

\begin{claim}
\label{claim:execs}
For every algorithm $V$ interacting with $\adv$, and every $\omega$-word $x \in \Sigma^\omega$,
there is a failure-free fair execution $E$ of $V$ such that~$x(E) = x$. 
\end{claim}

\begin{proof}
Since the system is fully asynchronous and $x$ is well-formed, $V$ admits a sequential execution,
constructed inductively as follows, where 
$x(j)$ denotes the symbol of $x$ at its $j$-th position, and
$\ell\_ind(x,j) = i$ if and only if $x(j) \in \Sigma_i$:

\begin{itemize}
\item Base. Process $p_{\ell\_ind(x,1)}$ executes Lines 1 to 3, where $x(1) \in \Sigma^<_{\ell\_ind(x,1)}$ 
is the invocation symbol it picks in Line 2.

\item Inductive step. For $k \geq 2$, the execution is extended as follows, depending on the symbol $x(k)$:

\begin{itemize}

\item if $x(k)$ is an invocation symbol. Process $p_{\ell\_ind(x,k)}$ executes Lines 1 to 3, where 
$x(k) \in \Sigma^<_{\ell\_ind(x,k)}$ is the invocation symbol it picks in Line 2;

\item if $x(k)$ is a response symbol. Process $p_{\ell\_ind(x,k)}$ executes Lines 4 to 6, where 
$x(k) \in \Sigma^>_{\ell\_ind(x,k)}$ is the response symbol it receives from the adversary in Line 4.

\end{itemize}

\end{itemize}

By construction, the input in the execution is $x$.
\end{proof}

\paragraph{Atomic snapshots.}
In some of the algorithms below, processes \emph{atomically} read all
entries of shared arrays using the well-known \emph{snapshot}
operation, that can be read/write wait-free implemented~\cite{AADGMS93}. 
Usage of snapshots is for simplicity, as, unless stated otherwise, 
the same results can be obtained through the weaker \emph{collect}
operation, possibly at the cost of more complex local computations.
Differently from snapshots,
this operation reads asynchronously, one by one, in an arbitrary order, the entries the shared array.

\section{Asynchronous Decidability}
\label{sec:asynch-decidability}

\subsection{Three definitions for distributed decidability}

First, we study three definitions for asynchronous decidability of distributed languages,
where it is assumed that processes can report only two possible values: $\yes$ or $\no$.
The first definition, strong decidability, is basically the distributed runtime verification
problem in~\cite{podc23}[Definition 3.1] adapted to our setting,
which in turn is motivated by the soundness and completeness requirements
for runtime verifications solutions considered in the literature (e.g.,~\cite{E-HF18, MB15, NFBB17}). 
The other two definitions, weak-all decidability and weak-one decidability, 
follow a B\"uchi-style $\omega$-words acceptance criteria~\cite{buchi},
and are suitable for eventual properties that can only be tested to the infinity (e.g. $WEC\_COUNT$ or $EC\_LED$).

As anticipated, in the decidability definitions, we focus on the fair failure-free executions of algorithms.
Given an execution $E$ of an algorithm~$V$, for each process $p$,
$\no(E,p)$ and $\yes(E,p)$ will denote the number of times $p$ reports $\no$ and $\yes$, respectively, in $E$.
Since $E$ is failure-free, at least one of $\no(E,p)$ and $\yes(E,p)$ is $\infty$, for every process $p$.

\begin{definition}[Strong Decidability]
An algorithm $V$ \emph{strongly decides} a language $L$ if in every execution $E$,
$$x(E) \in L \iff \forall p, \no(E,p) = 0.$$

Alternatively, we say that $L$ is \emph{strongly decidable}.
The class of strongly decidable languages is denoted $\sd$.
\end{definition}

\begin{definition}[Weak-All Decidability]
An algorithm $V$ \emph{weakly-all decides} a language $L$ if in every execution $E$,
$$x(E) \in L \iff \forall p, \no(E,p) < \infty.$$

Alternatively, we say that $L$ is \emph{weakly-all decidable}.
The class of weakly-all decidable languages is denoted $\wad$.
\end{definition}

\begin{definition}[Weak-One Decidability]
An algorithm $V$ \emph{weakly-one decides} a language $L$ if in every execution $E$,
$$x(E) \in L \iff \exists p, \no(E,p) < \infty.$$

Alternatively, we say that $L$ is \emph{weakly-one decidable}.
The class of weakly-one decidable languages is denoted $\wod$.
\end{definition}

\paragraph{About the restriction to failure-free executions:} Since blocks of code of algorithms interacting with~$\adv$ 
are wait-free and the system is asynchronous, 
even in failure-free executions processes are ``forced'' to make decisions (i.e., reporting values) without waiting
to ``hear'' from other processes, which is arguably the main challenge in 
asynchronous fault-tolerant systems. 
Therefore, focusing on failure-free executions is just for simplicity. 

\subsection{Basic properties}

This section shows stability properties of algorithms for the three decidability notions,
and that these properties imply that $\wad$ and $\wod$ are actually equivalent,
and hence they are just defined as weak decidability ($\wde$). 
Also, it is shown that $\sd$ is included in $\wde$.

\begin{lemma}
\label{lemma:basic-sd}
For every $L \in \sd$, there is an algorithm that strongly decides $L$ and satisfies
the following property, in every execution $E$:
if $x(E) \notin L$, eventually every process always reports~\no.
\end{lemma}

\begin{proof}
Let $V$ be an algorithm that strongly decides~$L$.
Algorithm $V$ has the generic structure described above. 
Without loss of generality, we assume that in Line~\ref{Gen6} of each $V_i$,
$p_i$ reports a value in the \emph{last} step of that block of code.
Consider the algorithm $W$ obtained by modifying Line~\ref{Gen6} of each local algorithm $V_i$ of $V$ as it appears in Figure~\ref{fig-basic-sd}.
The idea is that once a process is to report \no in $V$, in $W$ it sets a shared variable to remember this fact,
and reports \no in every subsequent iteration.

\begin{figure}[ht]
\centering{ \fbox{
\begin{minipage}[t]{150mm}
\scriptsize
\renewcommand{\baselinestretch}{2.5} \resetline
\begin{tabbing}
aaaa\=aaa\=aaa\=aaa\=aaa\=aaa\=aaa\=\kill 

{\bf Additional shared variables in $W$:}\\

$~~$  $FLAG:$ read/write register initialized to $\false$\\ \\

{\bf Modified Local algorithm $W_i$ for process $p_i$:}\\

 \> {\bf while} {\sf true} {\bf do}\\

\line{BSD1} \>\> Non-deterministically pick an invocation symbol $v_i \in \Sigma^<_i$\\

\line{BSD2} \>\> Block of code in Line~\ref{Gen2} of $V_i$\\

\line{BSD3} \>\> Send $v_i$ to the adversary $\adv$\\

\line{BSD4} \>\> Receive response symbol $w_i \in \Sigma^>_i$ from the adversary $\adv$\\

\line{BSD5} \>\> Block of code in Line~\ref{Gen5} of $V_i$\\

\line{BSD6} \>\> $d_i \rightarrow$ value in Line~\ref{Gen6} of $V_i$ that is $p_i$ to report\\

\>\> {\bf if} $FLAG.read() == \true$ {\bf then} {\bf report} $\no$\\

\>\> {\bf else}\\

\>\>\> {\bf if} $d_i == \no$ {\bf then} $FLAG.write(\true)$\\

\>\>\> {\bf report} $d_i$

\end{tabbing}
\end{minipage}
  }
\caption{From $V$ to $W$ in proof of Lemma~\ref{lemma:basic-sd}.}
\label{fig-basic-sd}
}
\end{figure}

Consider any execution $E_W$ of $W$. Note that we can obtain an execution $E_V$ of $V$
by replacing in $E_W$ every local computation of each process $p_i$ corresponding to Line~\ref{BSD6} of $W$ with
the corresponding local computation in Line~\ref{Gen6} of $V$,
namely, $p_i$ simply reports $d_i$ with no further computation.
Note that the inputs to both executions, $x(E_W)$ and $x(E_V)$, are the same. 
By definition of strong decidability, if $x(E_V) \in L$, then no process ever reports $\no$ in $E_V$,
and hence no process ever reports $\no$ in $E_W$.
If $x(E_V) \notin L$, there is a process $p_i$ that reports at least one time $\no$ in $E_V$, by definition
of strong decidability, hence $p_i$ sets $FLAG$ to $true$ in $E_W$,
which implies that, eventually every process reports $\no$ forever, since executions are fair.
Thus, $W$ strongly decides $L$, and has the desired stability property.
\end{proof}

\begin{lemma}
\label{lemma:basic-wad}
For every $L \in \wad$, there is an algorithm that weakly-all decides $L$ and satisfies
the following property, in every execution $E$:
if $x(E) \notin L$, every process reports $\no$ infinitely often.
\end{lemma}

\begin{proof}
Let $V$ be an algorithm that weakly-all decides~$L$.
Without loss of generality, we assume that in Line~\ref{Gen6} of each $V_i$,
$p_i$ reports a value in the \emph{last} step of that block of code.
Consider the algorithm $W$ obtained by modifying Line~\ref{Gen6} of 
each local algorithm $V_i$ of $V$ as shown in Figure~\ref{fig-basic-wad}.
In $W$, $p_i$ records in a new shared array $C$ the number of times it has obtained $\no$ from $V$
so far, then reads all entries of $C$, and finally
it reports $\no$ if there is an entry with a larger value than $p_i$ was aware of in the previous iteration, 
otherwise it reports $\yes$.

\begin{figure}[ht]
\centering{ \fbox{
\begin{minipage}[t]{150mm}
\scriptsize
\renewcommand{\baselinestretch}{2.5} \resetline
\begin{tabbing}
aaaa\=aaa\=aaa\=aaa\=aaa\=aaa\=aaa\=\kill 

{\bf Additional shared variables in $W$:}\\

$~~$  $C[1,\hdots,n]:$ shared array of read/write registers, each initialized to $0$\\ \\

{\bf Modified Local algorithm $W_i$ for process $p_i$:}\\

 \> $prev_i[1, \hdots, n] \leftarrow [0, \hdots, 0]$ \%\% additional local variable \\

 \> {\bf while} {\sf true} {\bf do}\\

\line{BWAD1} \>\> Non-deterministically pick an invocation symbol $v_i \in \Sigma^<_i$\\

\line{BWAD2} \>\> Block of code in Line~\ref{Gen2} of $V_i$\\

\line{BWAD3} \>\> Send $v_i$ to the adversary $\adv$\\

\line{BWAD4} \>\> Receive response symbol $w_i \in \Sigma^>_i$ from the adversary $\adv$\\

\line{BWAD5} \>\> Block of code in Line~\ref{Gen5} of $V_i$\\

\line{BWAD6} \>\> $d_i \leftarrow$ value in Line~\ref{Gen6} of $V_i$ that is $p_i$ to report\\

\>\> {\bf if} $d_i == \no$ {\bf then} $C[i].write(prev[i]+1)$\\

\>\> $snap_i \leftarrow \snap(C)$\\

\>\> {\bf if} $\exists j, snap_i[j] > prev_i[j]$ {\bf then} {\bf report} $\no$\\

\>\> {\bf else} {\bf report} $\yes$\\

\>\> $prev_i \leftarrow snap_i$

\end{tabbing}
\end{minipage}
  }
\caption{From $V$ to $W$ in proof of Lemma~\ref{lemma:basic-wad}.}
\label{fig-basic-wad}
}
\end{figure}

Consider any execution $E_W$ of $W$. As in the proof of Lemma~\ref{lemma:basic-sd},
observe that we obtain an execution $E_V$ of $V$
by replacing every local and shared computations of each process $p_i$ corresponding to Line~\ref{BWAD6} of $W$ with
the corresponding local computation in Line~\ref{Gen6} of $V$
(namely, $p_i$ simply reports $d_i$).
Note that $x(E_W) = x(E_V)$.
By definition of weakly-all decidability, if $x(E_V) \in L$, then every process reports $\no$ only finitely many times in $E_V$,
and hence eventually all values in $C$ stabilize in $E_W$, from which follows that every process reports $\no$ finitely many times in $E_W$.
Now, if $x(E_V) \notin L$, there is a process $p_i$ that reports $\no$ infinitely many times in $E_V$, by definition
of weakly-all decidability. This implies that $C[i]$ never stabilizes in $E_W$,
and hence every process infinitely often reads that $C[i]$ is increasing (recall that $E_W$ is fair), 
and as a consequence all processes report $\no$ infinitely often in $E_W$. 
Thus, $W$ weakly-one decides $L$, and has the desired property.
\end{proof}

\begin{lemma}
\label{lemma:basic-wod}
For every $L \in \wod$, there is
an algorithm that weakly-one decides $L$ and satisfies
the following property, in every execution $E$:
if $x(E) \in L$, eventually every process always reports $\yes$.
\end{lemma}

\begin{proof}
Let $V$ be an algorithm that weakly-one decides~$L$.
Without loss of generality, we assume that in Line~\ref{Gen6} of each $V_i$,
$p_i$ reports a value in the \emph{last} step of that block of code.
Consider the algorithm $W$ obtained by modifying Line~\ref{Gen6} of 
each local algorithm $V_i$ of $V$ as shown in Figure~\ref{fig-basic-wod}.
In $W$, each process $p_i$ records in a new shared array $C$ the number of times it has obtained $\no$ from $V$
so far, then reads all entries of $C$, and
it reports $\yes$ if there is an entry whose value has not changed, otherwise it reports $\no$.

\begin{figure}[ht]
\centering{ \fbox{
\begin{minipage}[t]{150mm}
\scriptsize
\renewcommand{\baselinestretch}{2.5} \resetline
\begin{tabbing}
aaaa\=aaa\=aaa\=aaa\=aaa\=aaa\=aaa\=\kill 

{\bf Additional shared variables in $W$:}\\

$~~$  $C[1,\hdots,n]:$ shared array of read/write registers, each initialized to $0$\\ \\

{\bf Modified Local algorithm $W_i$ for process $p_i$:}\\

 \> $prev_i[1, \hdots, n] \leftarrow [0, \hdots, 0]$ \%\% additional local variable \\

 \> {\bf while} {\sf true} {\bf do}\\

\line{BWOD1} \>\> Non-deterministically pick an invocation symbol $v_i \in \Sigma^<_i$\\

\line{BWOD2} \>\> Block of code in Line~\ref{Gen2} of $V_i$\\

\line{BWOD3} \>\> Send $v_i$ to the adversary $\adv$\\

\line{BWOD4} \>\> Receive response symbol $w_i \in \Sigma^>_i$ from the adversary $\adv$\\

\line{BWOD5} \>\> Block of code in Line~\ref{Gen5} of $V_i$\\

\line{BWOD6} \>\> $d_i \leftarrow$ value in Line~\ref{Gen6} of $V_i$ that is $p_i$ to report\\

\>\> {\bf if} $d_i == \no$ {\bf then} $C[i].write(prev[i]+1)$\\

\>\> $snap_i \leftarrow \snap(C)$\\

\>\> {\bf if} $\exists j, snap_i[j] == prev_i[j]$ {\bf then} {\bf report} $\yes$\\

\>\> {\bf else} {\bf report} $\no$\\

\>\> $prev_i \leftarrow snap_i$

\end{tabbing}
\end{minipage}
  }
\caption{From $V$ to $W$ in proof of Lemma~\ref{lemma:basic-wod}.}
\label{fig-basic-wod}
}
\end{figure}

Consider any execution $E_W$ of $W$. We obtain an execution $E_V$ of $V$
by replacing every local and shared computations of each process $p_i$ corresponding to Line~\ref{BWOD6} of $W$ with
the corresponding local computation in Line~\ref{Gen6} of $V$
(namely, $p_i$ simply reports $d_i$ with out any further computations). We have that $x(E_W) = x(E_V)$.
By definition of weakly-one decidability, if $x(E_V) \in L$, then there is a process $p_i$ that reports $\no$ finitely many times in $E_V$,
and hence eventually $C[i]$ stabilize in $E_W$, from which follows that every process reports $\no$ finitely many times in $E_W$.
If $x(E_V) \notin L$, all processes report $\no$ infinitely many times in $E_V$, by definition
of weakly-all decidability, which implies that no entry of $C$ ever stabilizes in $E_W$,
and hence all processes report $\no$ infinitely often in $E_W$. 
Thus, $W$ weakly-one decides $L$, and has the desired stability property.
\end{proof}

The previous lemmas imply:

\begin{theorem}
\label{theo:contentions}
$\sd \subseteq \wad = \wod$.
\end{theorem}

\begin{proof}
For any $L \in \sd$, consider an algorithm $V$ that strongly decides $L$ with the property stated in Lemma~\ref{lemma:basic-sd}.
In every execution $E$ of $V$, we have that if $x(E) \in L$, then $\no(E,p) = 0$, for every process $p$, 
and if  $x(E) \notin L$, $\no(E,p) = \infty$, for every process $p$. Thus, $V$ weakly-all decides $L$, and hence $\sd \subseteq \wad$.

For any $L \in \wad$, consider an algorithm $V$ that weakly-all decides $L$ with the property stated in Lemma~\ref{lemma:basic-wad}.
In every execution $E$ of $V$, if $x(E) \in L$, then for every process $p$, $\no(E,p) < \infty$, 
and if $x(E) \notin L$, $\no(E,p) = \infty$, for every process $p$. 
Thus, $V$ weakly-one decides $L$, and hence $\wad \subseteq WOD$.

For any $L \in WOD$, consider an algorithm $V$ that weakly-one decides $L$ with the property stated in Lemma~\ref{lemma:basic-wod}.
In every execution $E$ of $V$, if $x(E) \in L$, then for every process $p$, $\no(E,p) < \infty$, 
and if $x(E) \notin L$, $\no(E,p) = \infty$, for every process $p$. 
Thus, $V$ weakly-all decides $L$, and hence $\wod \subseteq \wad$.
\end{proof}

Given the equivalence above, we can alternatively define the clases $\wad$ and $\wod$ as follows:

\begin{definition}[Weak Decidability]
An algorithm $V$ \emph{weakly decides} a distributed language $L$ in every execution $E$,
\begin{itemize}
\item[]
$x(E) \in L \implies \forall p, \no(E,p) < \infty,$ 

\item[]
$x(E) \notin L \implies \forall p, \no(E,p) = \infty.$
\end{itemize}

Alternatively, we say that $L$ is \emph{weakly decidable}.
The class of weakly decidable languages is denoted~$\wde$.
\end{definition}

\section{Solvability results}

\subsection{Separation results}

We first show that linearizability and sequential consistency are in general neither 
strongly decidable nor weakly decidable.
The proof of the next impossibility result uses the same line of reasoning of that
in the proof of Theorem 5.1 in~\cite{podc23}, that exploits real-time order of events, that are unaccessible to the processes.

Given an algorithm $V$, two of its executions $E$ and $E'$ are \emph{indistinguishable to $p$}, 
denoted $E \equiv_p E'$, if $p$ passes through the same sequence of 
local states in both executions. If $E$ and $E'$ are indistinguishable to every process, 
we just say that they are \emph{indistinguishable}, denoted~$E \equiv E'$.
{The next proof, and others, use the fact that 
it could be $x(E) \neq x(E')$, despite $E \equiv E'$.}

\begin{lemma}
\label{lemma:no-wd-lin-sc}
$LIN\_REG, SC\_REG \notin \wde$.
\end{lemma}

\begin{proof}
We focus on the case $n=2$, but the argument below can be extended to any $n$.
By contradiction, suppose that there is an algorithm $V$ that weakly decides $LIN\_REG$.
Algorithm $V$ has the structure described in Figure~\ref{fig-generic-algo}. 
Consider the following execution $E$ of $V$, where $p_1$ and $p_2$ execute the loop iterations ``almost synchronously''
as described next. For every $r \geq 1$, their $r$-th iterations execute in the following order: 
\begin{enumerate}
\item $p_1$ picks $<^r_1$ (i.e. invocation to $write(r)$ by $p_1$) in Line~\ref{Gen1} and executes its computations in Line~\ref{Gen2} until completion.
\item $p_2$ picks $<_2$ (i.e. invocation to $read()$ by $p_2$) in Line~\ref{Gen1} and executes its computations in Line~\ref{Gen2} until completion.
\item $p_1$ sends $<^r_1$ to $\adv$ in Line~\ref{Gen3} and then receives $>_1$ from $\adv$ in Line~\ref{Gen4}.
\item $p_2$ sends  $<_2$ to $\adv$ in Line~\ref{Gen3} and then receives $>^r_2$ from $\adv$ in Line~\ref{Gen4} (namely, $p_2$ reads $r$ from the register).
\item $p_1$ executes its computations in Line~\ref{Gen5} and~\ref{Gen6} until completion.
\item $p_2$ executes its computations in Line~\ref{Gen5} and~\ref{Gen6} until completion.
\end{enumerate}

Observe that every prefix $x(E)$ is a linearizable history of register, where $p_1$ writes $r$ and immediately after $p_2$ reads $r$.
Thus, $x(E) \in LIN\_REG$, and hence, $\no(E,p_1), \no(E,p_2) <~\infty$.
Now consider the execution $F$ of $V$ obtained as $E$ except that items (3) and (4) above are swapped.
Note that $x(F)$ is not linearizable as $p_2$ reads $r$ before that value is written in the register. 
Hence, $x(F) \notin LIN\_REG$. Also, note that $p_1$ and $p_2$ cannot distinguish between the two executions (i.e., $E \equiv F$)
as the events in Lines~\ref{Gen3} and~\ref{Gen4} are local, hence it is impossible for them to know the order they are executed.
Thus, in $F$, $p_1$ and $p_2$ report the same sequence of values as in $E$, and hence $\no(F,p_1), \no(F,p_2)<~\infty$,
which is a contradiction as $x(F) \notin LIN\_REG$ and $V$ supposedly weakly decides $LIN\_REG$. 
Therefore, $LIN\_REG \notin~\wde$.

It is not difficult to verify that the very same argument proves that $SC\_REG \notin \wde$. 
The lemma follows.
\end{proof}

The previous lemma and Theorem~\ref{theo:contentions} imply:

\begin{corollary}
\label{coro:no-sd-lin-sc}
$LIN\_REG, SC\_REG \notin \sd$.
\end{corollary}

We now argue that the eventual counters defined above are not strongly decidable, 
but the weak version is weakly decidable. 

\begin{lemma}
\label{lemma:no-sd-evc}
$WEC\_COUNT, SEC\_COUNT \notin \sd$.
\end{lemma}

\begin{proof}
We focus on the case $n=2$, but the argument below can be extended to any $n$.
By contradiction, suppose that there is an algorithm $V$ that strongly decides $WEC\_COUNT$.
Consider the $\omega$-word $x$ that corresponds to a \emph{sequential} history
where $p_1$ executes $add()$,
and then $p_2$ and $p_1$ alternatively execute infinitely many $read()$ operations returning 0.
Specifically, $x$~is:
$$<^+_1 \hbox{ } >^{\,}_1 \hbox{    } <^{\,}_2 \hbox{ }  >^0_2 \hbox{    } <^{\,}_1 \hbox{ }  >^0_1
 \hbox{    }  <^{\,}_2 \hbox{ }  >^0_2 \hbox{    }  <^{\,}_1 \hbox{ }  >^0_1 \hbox{    } \hdots$$

Clearly $x \notin WEC\_COUNT$. By Claim~\ref{claim:execs},
there is an execution $E$ of $V$ such that $x(E) = x$.
The proof of Claim~\ref{claim:execs} shows that we can assume that in $E$ each process
atomically executes Lines~\ref{Gen1}-~\ref{Gen3} and Lines~\ref{Gen4}-~\ref{Gen6}, respectively, in every iteration of the loop.
Since $x \notin WEC\_COUNT$, at least one of $p_1$ and $p_2$ reports $\no$ in $E$.
Let $F$ be the shortest (finite) prefix of $E$ in which a process reports $\no$.
Without loss of generality, assume that such process reporting $\no$ is $p_2$.
Thus, at the end of $F$, $p_2$ executes its block of code corresponding to Line~\ref{Gen6}, reporting $\no$
(for the first time in $F$).
Let us consider the finite input $x(F)$ in prefix $F$ (namely, the projection of symbols in $\Sigma$).
Observe that $x(F)$ is a finite prefix of $x$ that ends with $<^{\,}_2 \hbox{ }  >^0_2$.
Consider the $\omega$-word $x'$:
$$x(F) \hbox{    }  <^{\,}_1 \hbox{ }  >^1_1 \hbox{    }  <^{\,}_2 \hbox{ }  >^1_2 \hbox{    }  <^{\,}_1 \hbox{ }  >^1_1 \hbox{    }  <^{\,}_2 \hbox{ }  >^1_2 \hdots$$
Note that $x' \in WEC\_COUNT$. By Claim~\ref{claim:execs},
there is an execution $E'$ of $V$ such that $x(E) = x'$.
Following the proof of Claim~\ref{claim:execs}, we can see that $F$ is prefix of $E'$,
and hence $p_2$ reports $\no$ in~$E'$.
But this is a contradiction as  $x' \in WEC\_COUNT$.
Therefore, $WEC\_COUNT \notin \sd$. 

Finally, since $SEC\_COUNT \subset WEC\_COUNT$, we also have that $SEC\_COUNT \notin \sd$.
\end{proof}

Using a similar line of reasoning, we can provide ad hoc arguments showing that none of the languages
$LIN\_LED, SC\_LED$ and $EC\_LED$ does not belong to $\wde$ or $\sd$, and that $SEC\_COUNT$ is not in $\wde$.
Instead, we will prove these impossibility results through a characterization in the next subsection.

\begin{lemma}
\label{lemma:wd-evc}
$WEC\_COUNT \in \wde$.
\end{lemma}

\begin{proof}
We argue that the read/write algorithm $V$ in Figure~\ref{fig-algo-wec-count} weakly decides $WEC\_COUNT$.
In the algorithm, before interacting with $\advt$,
process $p_i$ announces in $M[i]$ if in its current iteration-loop it sends to $\advt$ an invocation to $inc()$.
After interacting with $\advt$, in the block in Line~\ref{EC5}, $p_i$ reads all increments announced in $INCS$ so far,
and records the returned value in $w_i$, in case it is a response to $read()$.
In the block in Line~\ref{EC6}, $p_i$ reports $\no$ if $flag_i$ encodes that 
$p_i$ already detected that one of the first two properties in the definition of $WEC\_COUNT$ has been violated;
next $p_i$ checks if in the current iteration $p_i$ witnesses that one of the first two properties
does not hold, and if so, it encodes that in $flag_i$ and reports $\no$;
then $p_i$ checks if the current iteration violates the third property in the definition of $WEC\_COUNT$,
and if so it only reports $\no$;
and finally, if none of the previous cases hold, $p_i$ reports $\yes$.

\begin{figure}[ht]
\centering{ \fbox{
\begin{minipage}[t]{150mm}
\scriptsize
\renewcommand{\baselinestretch}{2.5} \resetline
\begin{tabbing}
aaaa\=aaa\=aaa\=aaa\=aaa\=aaa\=aaa\=\kill 

{\bf Shared variables in $V$:}\\

$~~$  $INCS[1,\hdots,n]:$ shared array of read/write registers, each initialized to $0$\\ \\

{\bf Local algorithm $V_i$ for process $p_i$:}\\

 \> $prev\_read_i \leftarrow 0$\\
 
 \> $prev\_incs_i \leftarrow 0$\\
 
  \> $count_i \leftarrow 0$\\
 
 \> $flag_i \leftarrow \false$\\

 \> {\bf while} {\sf true} {\bf do}\\

\line{EC1} \>\> Non-deterministically pick an invocation symbol $v_i \in \Sigma^<_i$\\

\line{EC2} \>\> {\bf if} $v_i$ is an invocation to $inc()$ {\bf then} \\

 \>\>\> $count \leftarrow count + 1$\\
 
  \>\>\> $INCS[i].write(count)$ \\

\line{EC3} \>\> Send $v_i$ to the adversary $\adv$\\

\line{EC4} \>\> Receive response symbol $w_i \in \Sigma^>_i$ from the adversary $\adv$\\

\line{EC5} \>\> $snap_i \leftarrow \snap(INCS)$\\

\>\> $curr\_incs_i \leftarrow snap_i[1] + \hdots + snap_i[n]$\\

\>\> {\bf if} $w_i$ is a response to $read()$ {\bf then} $curr\_read_i \leftarrow$ return value in $w_i$\\

\line{EC6} \>\> {\bf if} $flag_i == \true$ {\bf then}\\ 

\>\>\> {\bf report} $\no$\\

\>\> {\bf elseif} $curr\_read_i < snap_i[i] \, \vee \, curr\_read_i < prev\_read_i$ {\bf then}\\

\>\>\> $flag_i \leftarrow \true$\\

\>\>\> {\bf report} $\no$\\

\>\> {\bf elseif} $curr\_read_i \, != curr\_incs_i \, \vee \, prev\_incs_i < curr\_incs_i$ {\bf then}\\

\>\>\> {\bf report} $\no$\\

\>\> {\bf else}\\ 

\>\>\> {\bf report} $\yes$\\

 \>\> $prev\_read_i \leftarrow curr\_read_i$\\
 
 \>\> $prev\_incs_i \leftarrow curr\_incs_i$   

\end{tabbing}
\end{minipage}
  }
\caption{Weakly deciding $WEC\_COUNT$.}
\label{fig-algo-wec-count}
}
\end{figure}

Let $E$ be any execution of $V$. We have two cases:

\begin{itemize}

\item $x(E) \in WEC\_COUNT$. Eventually every process always picks $read$ operations in Line~1,
and hence eventually all values in array $INCS$ stabilize. This implies that $curr\_inc_i$ of each process $p_i$ stabilizes too.
Since we assume fair executions of $V$, every invocation to $inc$
in $x(E)$ eventually is reflected in $INCS$. 
Moreover, since $x(E) \in WEC\_COUNT$, $read$ operations monotonically increase returned values,
and each is at least the the number of previous $inc$ operations of the process.
Thus, $flag_i$ of each process is never set to $\true$.
These observation imply that, for all processes, the first two clauses in Line~\ref{EC6} are never satisfied in the execution,
and eventually the third clause is never satisfied.
Thus, every process reports $\no$ only finitely many times in $E$.

\item $x(E) \notin WEC\_COUNT$. First observe that if $x(E)$ has a $read$ operation
of a process $p_i$ that does not satisfy one of the first two properties of the definition of the weakly-eventual consistent counter,
then this process eventually sets its variable $flag_i$ to $\true$, and hence the process reports $\no$ infinitely many times in $E$.
And if $x(E)$ has infinitely many $inc$ operations, then the values in $INCS$ never stabilize, 
and thus the third clause in Line~\ref{EC6} is satisfied infinitely many
times for all processes, which implies that every process reports $\no$ infinitely many times in~$E$.
The case that remains to be considered is that $x(E)$ satisfies the first two properties 
of the definition of the weakly-eventual consistent counter, and it has finitely many $inc$ operations.
Since, $x(E) \notin WEC\_COUNT$, it must be that every infinite suffix of $x(E)$ with only $read$ operations,
has a $read$ operation that returns a value that is distinct from the number of $inc$ operations in $E$.
Using a similar reasoning as above, it can be argued that eventually the values in $INCS$ stabilize
and every $inc$ is reflected in $INCS$. 
Let $S$ denote the number of $inc$ operations in $E$.
The observations we have made imply that in $E$
there are infinitely many read operation that return a value distinct from $S$,
and hence there is at least one process that reports $\no$ infinitely many times in $E$.
\end{itemize}

From the arguments so far, we have that 

\begin{itemize}
\item[]
$x(E) \in L \implies \forall p, NO(E,p) < \infty,$ 
\item[]
$x(E) \notin L \implies \exists p, NO(E,p) = \infty.$
\end{itemize}

Thus, $V$ weakly-all decides $WEC\_COUNT$, hence $WEC\_COUNT \in \wad$.
By Lemma~\ref{lemma:basic-wad}, $V$ can be transformed into 
an algorithm that weakly decides $WEC\_COUNT$, from which follows that $WEC\_COUNT \in \wde$.
\end{proof}

From the previous results, we obtain the following separation result:

\begin{theorem}
\label{theo:separation-sd-wd}
$\sd \subset \wde$.
\end{theorem}

\subsection{Characterization of asynchronous decidability}
\label{sec:characterization}

We now present one of our main results,
a characterization of the languages that can be decided against $\adv$. 
The characterization considers a \emph{generic} decidability notion defined through a \emph{decidability predicate} $\sf P$.
The predicate describes a property that
reported values satisfy in any (fair failure-free) execution $E$ whose input $x(E)$ is in the language that is decided.
Remarkably, and differently from {the predicates in the definitions of $\sd$ and $\wde$},
predicate $\sf P$ might involve more than only $\yes$ and $\no$ report values (e.g., $\yes$, $\no$ and $\maybe$);
furthermore, $\sf P$ might allow infinitely many distinct report values.
Decidability with respect to $\sf P$ is stated as follows:

\begin{definition}[$\sf P$-decidability]
Given a decidability predicate $\sf P$, a language $L$ is $\sf P$-\emph{decidable} if and only if there exists an algorithm $V$
such that in every execution $E$, $x(E) \in L \Longleftrightarrow {\sf P}(E) = \true$.
\end{definition}

\begin{definition}
Let $x_1 \shuffle \hdots \shuffle x_m$ denote
the \emph{shuffle} of words $x_1, \hdots, x_m$, namely, the set with all interleavings of $x_1, \hdots, x_m$.
\end{definition}

\begin{definition}[Real-Time Oblivious Languages]
\label{def:rt-oblivious}
A language $L$ is \emph{real-time oblivious} if for every $\alpha \beta \in L$ with $\alpha$ finite,
$\alpha' \beta \in L$, for every $\alpha' \in \alpha|1 \shuffle \hdots \shuffle \alpha|n$.
\end{definition}

Intuitively, a real-time oblivious language describes a distributed service where the sequence of responses of a process
\emph{do not} depend of previous invocations and responses of other processes, because 
its responses remain correct in all possible interleavings.
It is not difficult to verify that $WEC\_COUNT$ is real-time oblivious, but $SEC\_COUNT$ is not due to
the fourth property in its definition.

{The proof of Theorem~\ref{theo:rt-oblivious} below follows a strategy similar to
impossibility proofs in the literature (e.g.,~\cite{FLP85}), 
where it is constructed a sequence of executions such that at least one process 
does not distinguish between every pair of consecutive executions in the sequence,
with the aim to argue the decisions in the first and last executions are somehow linked. 
The proof exploits indistinguishability implied by real-time order of events.
Roughly speaking, it constructs a sequence of executions $E_0, E_1, E_2, \hdots, E_{2x}$
such that, for each $0 \leq k \leq x-1$, 
(1) $E_k$ and $E_{k+1}$ are indistinguishable to all processes, 
but inputs in the executions are the different (due to real-time order unaccessible to the processes),
hence the report values in both executions are the same, and 
(2) $E_{k+1}$ and $E_{k+2}$ might be distinguishable to some processes,
but inputs in the executions are the same, hence in both executions
the decidability predicate $\sf P$ either holds or does not hold.
In this way, it can be concluded that $x(E_0) \in L \iff x(E_{2x}) \in L$.
}

\begin{theorem}
\label{theo:rt-oblivious}
For every decidability predicate $\sf P$, 
$L$ is $\sf P$-decidable $\implies L \hbox{ is real-time oblivious}$.
\end{theorem}

\begin{proof}
Let $V$ be an algorithm that $\sf P$-decides $L$. 
Consider any $x = \alpha \beta \in L$ with $\alpha$ finite. 
By Claim~\ref{claim:execs}, there is an execution $E$ of $V$ such that $x = x(E)$.
Consider any  $\alpha' \in \alpha|1 \shuffle \hdots \shuffle \alpha|n$.
We will argue that there is an execution $E'$ of $V$ such that $x(E') = \alpha' \beta \in L$.
Below, $\ell(y,y')$ denotes the longest common prefix of $y$ and $y'$.

If $|\ell(\alpha, \alpha')|=|\alpha|$, then $\alpha = \alpha'$, and hence $\alpha' \beta \in L$.
Thus, suppose $|\ell(\alpha, \alpha')| < |\alpha|$. The proof of the theorem is based on the next claim:

\begin{claim}
If $|\ell(\alpha, \alpha')| < |\alpha|$, there is an execution $E''$ of $V$ such that $x(E'') \in L$
and $x(E'') = \alpha'' \beta$ with $\alpha'' \in \alpha|1 \shuffle \hdots \shuffle \alpha|n$ such that 
$|\ell(\alpha'', \alpha')| \geq |\ell(\alpha, \alpha')|+1$.
\end{claim}

\begin{proof}[Proof of claim]
Let $\sigma = \ell(\alpha, \alpha')$.
We have $\alpha = \sigma v \tau$ and $\alpha' = \sigma v' \tau'$ for some
symbols $v,v'$ of $\Sigma$ and words $\tau, \tau'$ over $\Sigma$, such that $v \neq v'$ and $\tau \neq \tau'$.
For sake of simplicity, let us assume that $v$ and $v'$ appear only once in $\alpha \beta$ 
and $\alpha' \beta$, respectively.~\footnote{Alternatively, we can mark the symbols of a string with their positions 
in it in order to make them unique.}
Observe that $\alpha' \in \sigma|1, \hdots, \sigma|n$ implies 
$v'$ appears somewhere in $\tau$ and $v$ appears somewhere in $\tau'$.
Let $v \in \Sigma_j$ and $v' \in \Sigma_i$.
We argue that $i \neq j$: we reach a contradiction if $i = j$, 
because then $v$ and $v'$ appear in opposite orders in $\alpha$ and $\alpha'$,
which implies that $\alpha|i \neq \alpha'|i$, contradicting $\alpha' \in \sigma|1, \hdots, \sigma|n$.
Using a similar reasoning, we argue that in the substring of $\alpha$ between $v$ and $v'$, 
there is no $u \in \Sigma_i$: we reach a contradiction if there is such symbols $u$, because
then $u$ appears somewhere in $\tau'$, which contradicts $\alpha' \in \sigma|1, \hdots, \sigma|n$,
as $u$ and $v'$ appear in opposite orders in $\alpha$ and $\alpha'$. 

We now reason about the execution $E$. 
Let us consider the events in $E$ that send or receive symbols $v$ and $v'$ 
(corresponding to Lines~\ref{Gen3} or~\ref{Gen4} in the generic algorithm in Figure~\ref{fig-generic-algo}). 
For simplicity, those events will be denoted $v$ and $v'$. 
The discussion above implies that in $E$:

\begin{enumerate}
\item event $v$ appears before event $v'$,

\item $v$ and $v'$ are events of $p_j$ and $p_i$, respectively, 

\item $p_i \neq p_j$, and

\item there is no other send/receive event of $p_i$ between~$v$ and~$v'$.
\end{enumerate}

Let $H$ be the subsequence of $E$ that starts at $v$ and ends at $v'$.
Observe that $H$ might have steps of $p_i$ corresponding to local or shared memory computations 
(steps that are part of blocks of code in Lines~\ref{Gen1},~\ref{Gen2},~\ref{Gen5} or~\ref{Gen6}).
Consider the execution $F$ obtained by ``moving back'' all those $p_i$'s steps right before $v$.
Observe that asynchrony allows such modification of $E$. 
Moreover, note that indeed $F$ is an execution of $V$, but the state of processes after $v$ 
might differ, due to the fact that some shared memory computations of $p_i$ that appear after $v$ in $E$,
now appear before $v$ in $F$. 
Thus, the values reported in $F$ might be different than the values reported in $E$.
However, the relative order of send/receive events remains the same, thus $x(F) = x(E)$.
Since $x(E) \in L$ and $V$ is assumed to $\sf P$-decide $L$,
it must be that ${\sf P}(F)$ is true.

To conclude the proof, we now modify $F$. Observe that in $F$ there are no steps of $p_i$ between 
$v$ and $v'$. Due to asynchrony, we can move back $v'$ right before $v$. Let $E''$ denote the modified execution.
Note that $E''$ is indeed an execution of $V$ because only a local step of $p_i$ was modified, 
whose exact time of occurrence is immaterial for all processes, including $p_i$ itself. 
Furthermore, all processes pass through the same sequence of states in both executions, 
namely, $F \equiv E'' $, and hence processes report the same values.
As already stated, ${\sf P}(F)$ is true, which implies that ${\sf P}(E'')$ is true as well.
However, $x(E'') \neq x(F)$, as $v$ and $v'$ appear in opposite orders in $x(E'')$ and $x(F)$. 
Since $V$ is assumed to $\sf P$-decide $L$, it must be that $x(E'') \in L$.
Now, due to the single modification made to obtain $E''$ from $F$, 
there must exist $\alpha'' \in \alpha|1 \shuffle \hdots \shuffle \alpha|n$
that is prefix of $x(E'')$. Note that $\sigma v'$ is prefix of $\alpha''$,
recalling that $\sigma = \ell(\alpha, \alpha')$.
Then, $\sigma v'$ is prefix of $\alpha'$ and $\alpha''$,
which implies that $|\ell(\alpha'', \alpha')| \geq |\ell(\alpha, \alpha')|+1$.
Therefore, $E''$ is an execution of $V$ that has the desired properties. 
The claim follows.
\end{proof}

To complete the proof, we repeatedly apply the previous claim a finite number of times to obtain a sequence of executions 
whose respective inputs belong to $L$, until we reach one whose input is precisely $\alpha' \beta$.
The theorem follows.
\end{proof}

It is possible to check that none of the languages $LIN\_REG, SC\_REG$ and $SEC\_COUNT$ 
is real-time oblivious, hence we can use Theorem~\ref{theo:rt-oblivious}, {instantiated with 
${\sf P}$ equal to the predicate in the definition of $\sd$ or $\wde$}, to reprove 
the impossibility results in Lemmas~\ref{lemma:no-wd-lin-sc} and~\ref{lemma:no-sd-evc} and Corollary~\ref{coro:no-sd-lin-sc}, 
except for the case of $WEC\_COUNT$.
Furthermore, it is easy to see that 
$LIN\_LED, SC\_LED$ and $EC\_LED$ are not real-time oblivious (see Appendix~\ref{app:languages}),
and thus they are neither strongly decidable nor weakly decidable, by Theorem~\ref{theo:rt-oblivious}:

\begin{corollary}
\label{coro:no-sd-lin-sc-2}
$LIN\_LED, SC\_LED, EC\_LED \notin \sd$
\end{corollary}

\begin{corollary}
\label{coro:no-wd-lin-sc}
$LIN\_LED, SC\_LED, EC\_LED \notin \wde$
\end{corollary}

Observe that Lemmas~\ref{lemma:no-sd-evc} and~\ref{lemma:wd-evc} do not contradict each other about $WEC\_COUNT$:
Lemma~\ref{lemma:wd-evc} and Theorem~\ref{theo:rt-oblivious} imply that $WEC\_COUNT$ is real-time oblivious,
but this property does not suffice to make $WEC\_COUNT$ strongly decidable, as shown in Lemma~\ref{lemma:no-sd-evc},
which is not a contradiction as Theorem~\ref{theo:rt-oblivious} is not a full characterization.

\paragraph{Relation between Theorem~\ref{theo:rt-oblivious} and~\cite{darv, opinions}.}
As discussed in Section~\ref{sec:related-work}, some models for distributed, asynchronous, crash-tolerant runtime verification assume that 
the verified system is static~\cite{opinions} or dynamic but does not change to the next 
state until a runtime verification phase completes~\cite{darv}. 
These assumptions yield strong lower bounds in those works, although they apply only to read/write algorithms.
In contrast, our setting not only assumes full asynchrony but also aims to verify properties
 that may involve real-time order constraints. 
 While the impossibility results in~\cite{darv, opinions}
  already highlight the difficulty of runtime verification, 
it was still unclear how hard this can be in general settings.  
Theorem~\ref{theo:rt-oblivious} proves that, under asynchrony, \emph{all} real-time sensitive properties 
(that is, non-real-time oblivious) are runtime \emph{unverifiable},
regardless of the number of report values, the choice of decidability predicate, or the computational 
 power of the base primitives.
This stands in sharp contrast to~\cite{darv, opinions}, where the absence of real-time constraints permits to
assign a \emph{finite} number~$k$ (\emph{alternation number}) to every property
such that at most $2k+4$ report values (\emph{opinions}) are needed to runtime verify the property,
in those restricted settings.

\section{Timed Adversaries and Predictive Decidability}
\label{sec:timed-adversaries}

Although strong decidability is arguably highly desirable, the results so far suggest that
there is little to do in this direction. 
In general, strong correctness conditions such as 
linearizability and sequential consistency are impossible (Corollaries~\ref{coro:no-sd-lin-sc} and~\ref{coro:no-sd-lin-sc-2}),
eventual versions of the counter and ledger objects are imposible too 
(Lemma~\ref{lemma:no-sd-evc} and Corollary~\ref{coro:no-sd-lin-sc-2}),
and moreover only weak distributed problems can be strongly or weakly decided (Theorem~\ref{theo:rt-oblivious}),
problems with no order constraints between operations of different processes (Definition~\ref{def:rt-oblivious}).

The root of these impossibility results is the big power 
the adversary $\adv$ has to ``mislead'' any algorithm
trying to runtime verify $\adv$'s current behavior. It turns out that such power
can be reduced considerably if $\adv$ is verified ``indirectly'', by means of an adversary $\advt$
obtained from~$\adv$, as shown recently in~\cite{podc23, rv24}. 
Actually, through this indirect verification, strong correctness conditions like
linearizability of \emph{every} object turns ``almost'' strongly decidable.
The key step in this direction is a simple
but powerful transformation that takes $\adv$ and produces a \emph{new distributed service} $\advt$
that ``wraps'' $\adv$ in simple read/write wait-free code that is executed before and after interacting with $\adv$.
The aim of this transformation is to limit the power of $\adv$
by adding information to its responses. 
This extra information, called \emph{view}, basically \emph{timestamps} responses 
of the new distributed service (adversary) $\advt$ that algorithms will  runtime verify.
In this sense, $\advt$ is a version of $\adv$ that produces \emph{timed} histories/words.

We refer the reader to~\cite{podc23} for a detailed discussion about 
the properties of~$\advt$.
In the following two sections
we just provide high-level ideas of the transformation, state its main properties,
and show the algorithm in~\cite{podc23} that almost strongly runtime verifies linearizability.

\subsection{The timed adversary $\advt$}

The transformation from $\adv$ to $\advt$ appears in Figure~\ref{fig-timed-adv}.
In $\advt$, each process $p_i$ keeps in a shared variable $M[i]$ the (unordered) multiset of invocations it has sent so far to $\adv$.
Before sending its current invocation $v_i$ (Line~\ref{AT3}), $p_i$ announces it in $M[i]$ (Line~\ref{AT2}),
and after obtaining a response $w_i$ from $\adv$ (Line~\ref{AT4}), it snapshots all entries in $M$,
storing the union of all of them in a set $view_i$, called \emph{view} (Lines~\ref{AT5}-~\ref{AT6}), 
finally returning the tuple $(w_i,view_i)$ (Line~\ref{AT7}).

\begin{figure}[ht]
\centering{ \fbox{
\begin{minipage}[t]{150mm}
\scriptsize
\renewcommand{\baselinestretch}{2.5} \resetline
\begin{tabbing}
aaaa\=aaa\=aaa\=aaa\=aaa\=aaa\=aaa\=\kill 

{\bf Shared variables in $\advt$:}\\

$~~$ $S[1,\hdots, n]:$ array of read/write registers, each initialized to $\emptyset$\\ \\

{\bf Local persistent variable of $p_i$:}\\

$~~$ $s_i \leftarrow \emptyset$\\ \\

{\bf Local algorithm for process $p_i$:}\\

 \> {\bf When} receive $v_i$ ($\in \Sigma^<_i)$ by $p_i$ {\bf do}\\

\line{AT1} \>\> $s_i \leftarrow s_i \cup \{v_i\}$\\

\line{AT2} \>\> $M[i].write(s_i)$\\

\line{AT3} \>\> Send $v_i$ to the adversary $\adv$\\

\line{AT4} \>\> Receive response symbol $w_i \in \Sigma^>_i$ from the adversary $\adv$\\

\line{AT5} \>\> $snap_i \leftarrow \snap(M)$\\

\line{AT6} \>\> $view_i \leftarrow snap_i[1] \cup \hdots \cup snap_i[n]$\\

\line{AT7} \>\> Send back $(w_i, view_i)$ to $p_i$

\end{tabbing}
\end{minipage}
  }
\caption{The timed adversary $\advt$.}
\label{fig-timed-adv}
}
\end{figure}

From now on, we consider algorithms that follow the generic structure in Figure~\ref{fig-generic-algo},
and interact in Lines~\ref{Gen3} and~\ref{Gen4} with the timed adversary $\advt$.
Thus, responses from $\advt$ are of the form $(w_i,view_i)$.
In any such algorithm $V$, the steps of $V$ and $\advt$ interleave in an execution:
when process $p_i$ sends an invocation $v_i$ to $\advt$ in Line~\ref{Gen3} of Figure~\ref{fig-generic-algo},
it then executes Lines~\ref{AT1} and~\ref{AT2} of Figure~\ref{fig-timed-adv},
and only then it sends $v_i$ to $\adv$ in Line~\ref{AT3}, and $p_i$'s computation continues
once it obtains a response $w_i$ from~$\adv$ in the next line;
when $p_i$ finally completes its code in $\advt$, it resumes its computation in $V$.

For simplicity, and without loss of generality, it is assumed that in every execution $p_i$ sends each
$v_i \in \Sigma_i$  to $\advt$ at most once (alternatively, each invocation symbol in $x(E)$ 
could be marked with its position in $x(E)$ to make it unique).

In an execution $E$ of $V$, the \emph{input} $x(E)$ is the $\omega$-word obtained by projecting
the invocations to and responses from $\advt$ in Lines~\ref{Gen3} and~\ref{Gen4}
of Figure~\ref{fig-generic-algo}, ignoring views in the responses.

Consider any execution $E$ of $V$. 
In $x(E)$, any operation $(v_i,w_i)$ of process $p_i$ has a view $view_i$ that 
$\advt$ sends back to $p_i$ together with~$w_i$ (i.e., $\advt$ sends back $(w_i, view_i)$).
By the design of $\advt$, $view_i$ contains the invocations of \emph{all} operations
 that \emph{precede} operation $(v_i,w_i)$ in $x(E)$, 
and \emph{some} of the operations that are \emph{concurrent} to $(v_i,w_i)$ (see the example in Figure~\ref{fig:execution}).
It turns out that this information in the views suffices 
for the processes to locally obtain a concurrent history 
that is not exactly $x(E)$ but is closely related to it (see Appendix~\ref{app:construction}).
Basically, the history that can be obtained from the views, denoted $x^\sim(E)$, is one in which the operations in $x(E)$ might  ``shrink''.
In the parlance of~\cite{podc23}, $x^\sim(E)$ is an \emph{sketch} of $x(E)$ as it only resembles it.
Importantly, $x^\sim(E)$ is the input of \emph{some} execution $E'$ of $V$ \emph{indistinguishable} from $E$ to \emph{all} processes, 
and hence $x^\sim(E)$ is indeed a possible behavior of~$\advt$. 
The main properties of $\advt$ and the sketch $x^\sim(E)$ obtained from the views it produces can be summarized
as follows (implied by the construction in Section~7 and  Lemma~7.4 in \cite{podc23}):

\begin{figure}[ht]
\begin{center}
\includegraphics[scale=0.275]{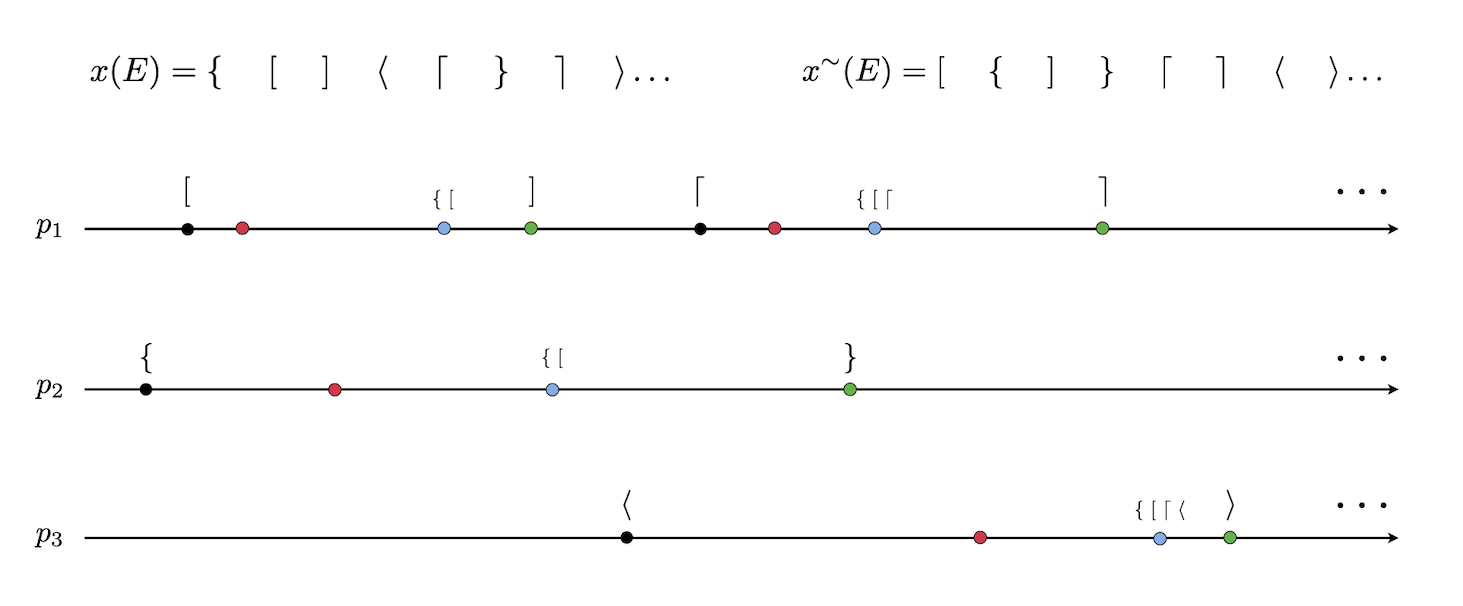}
\caption{\small The figure schematizes a prefix of an execution $E$ of a 3-process algorithm $V_O$ (Figure~\ref{fig-deciding-lin}) 
interacting with adversary $\advt$ (Figure~\ref{fig-timed-adv}).
It only shows send and receive events in $V_O$ (black and green circles, respecitely), and
write and snapshot events in $\advt$ (red and blue circles, respectively). 
Invocations sent and responses received by processes are indicated above the corresponding steps,
as well as the set of invocations obtained by snapshots; views in responses are omitted.
The history $x^\sim(E)$ constructed from the views (see Appendix~\ref{app:construction}) corresponds 
to the history that is obtained by ``moving forward'' each invocation to the next write,
and ``moving backward'' each response to the previous snapshot. 
Thus, history $x^\sim(E)$ is $x(E)$ where some operations might ``shrink''.
}
\label{fig:execution}
\end{center}
\end{figure}

\begin{theorem}
\label{theo:prev-work}
In every execution $E$ of an algorithm $V$ interacting with $\advt$:
\begin{enumerate}
\item 
\label{theo:prev-prop1}
$op \prec_{x(E)} op' \implies op \prec_{x^\sim(E)} op'$, and

\item 
\label{theo:prev-prop2}
there is an execution $E'$ of $V$ such that $E' \equiv E$ and $x(E') = x^\sim(E)$.
\end{enumerate}
\end{theorem}

\subsection{Predictive strong decidability}

As we did in Section~\ref{sec:dist-lang} for the register and the ledger objects, 
for any sequential object $O$, we can define the language $LIN\_O$
with every $\omega$-word such that each finte prefix is linearizable
with respect $O$.\footnote{The only assumption needed is that $O$ is \emph{total},
namely, in each of its states, every operation can be invoked.}

We would like to exploit the properties of the timed adversary $\advt$ in Theorem~\ref{theo:prev-work}
to runtime verify $LIN\_O$. The first step towards this direction is to recall that $\advt$ is nothing else than
a distributed service that ``wraps'' $\adv$, and argue that linearizability of 
$\adv$ and $\advt$ are linked. This is shown in Lemma~\ref{lemma:timed-adv-lin}
(implied by Lemma 7.2 in~\cite{podc23}),
where for simplicity $\adv$ and $\advt$ denote themselves the set of all possible histories 
that the distributed services can exhibit,
and $\adv$ (resp. $\advt$) being linearizable means that every finite prefix of each of its histories is linearizable,
with respect to a given sequential object $O$.

\begin{lemma}
\label{lemma:timed-adv-lin}
$\adv \hbox{ is linearizable } \Longleftrightarrow \advt \hbox{ is linearizable (ignoring views)}$.
\end{lemma}

\begin{figure}[ht]
\centering{ \fbox{
\begin{minipage}[t]{150mm}
\scriptsize
\renewcommand{\baselinestretch}{2.5} \resetline
\begin{tabbing}
aaaa\=aaa\=aaa\=aaa\=aaa\=aaa\=aaa\=\kill 

{\bf Shared Variables:}\\

$~~$  $M[1, \hdots, n]:$ shared array of read/write registers, each initialized to $\emptyset$\\ \\

{\bf Local algorithm for process $p_i$:}\\

 \> $s_i = \emptyset$\\

 \> {\bf while} {\sf true} {\bf do}\\

\line{VL1} \>\> Non-deterministically pick $v_i \in \Sigma^<_i$\\

\line{VL3} \>\> \%\% No communication is needed before sending $v_i$ to $\advt$\\

\line{VL3} \>\> Send $v_i$ to the adversary $\advt$\\

\line{VL4} \>\> Receive $(w_i, view_i)$, $w_i \in \Sigma^>_i$, from the adversary $\advt$\\

\line{VL5} \>\> $s_i \leftarrow s_i \cup \{ (v_i, w_i, view_i) \}$\\

\>\> $M[i].write(s_i)$\\

\>\> $snap_i \leftarrow \snap(M)$\\

\>\> $h_i \leftarrow$ finite history obtained through the triples in $snap_i$, as in~\cite{podc23}\\

\line{VL6} \>\> {\bf if} $h_i$ is linearizable w.r.t. $O$ {\bf then report} $\yes$\\

\>\> {\bf else} {\bf report} $\no$

\end{tabbing}
\end{minipage}
  }
\caption{An algorithm $V_O$ that predictively strongly decides $LIN\_O$.}
\label{fig-deciding-lin}
}
\end{figure}

The lemma implies that the only reason $\advt$ can exhibit a 
non-linearizable behavior is because~$\adv$ is not linearizable, and not because of the extra code in $\advt$.
This is the basis of the ``indirect'' verification strategy mentioned before.
As $\adv$ is a black-box, we will still assume that $\adv$ can exhibit any possible behavior,
as in previous sections, and hence $\advt$ can exhibit any possible behavior too.

Figure~\ref{fig-deciding-lin} contains algorithm $V_O$ in~\cite{podc23} that
almost strongly decides $LIN\_O$, where the reader is referred to Appendix~\ref{app:construction}
for the details how $h_i$ is locally computed in Line~\ref{VL6}. 
Below we use Theorem~\ref{theo:prev-work} to informally
explain the rational behind $V_O$,
in order to motivate a decidability notion,
predictive strong decidability, that formalizes ``almost strongly deciding''.~\footnote{It is already shown in~\cite{podc23}
that linearizability is not strongly decidable against $\advt$, hence a relaxation of strong decidability
is the most feasible alternative.}

Let $E$ be any execution of $V_O$. If $x(E) \notin LIN\_O$, 
then Theorem~\ref{theo:prev-work}(1) implies $x^\sim(E) \notin LIN\_O$:
intuitively, if $x(E)$ is not linearizable, then $x^\sim(E)$, where operations might ``shrink'', is not linearizable either.
Thus, there is a finte prefix $\alpha$ of $x^\sim(E)$ such that every prefix of $x^\sim(E)$
having $\alpha$ as a prefix is not linearizable. This holds due to the fact that
linearizability is a \emph{prefix-closed} property, which roughly speaking means that there is nothing that
can happen in the future that can ``fix'' a non-linearizable prefix.    
Therefore, as $E$ is fair, eventually every process always computes in Line~\ref{VL6}
a prefix of $x^\sim(E)$ that is not linearizable~\footnote{Due to asynchrony, 
$x_i$ might not be exactly a prefix, but a non-linearizable history that is ``similar'' to a non-linearizable prefix.},
hence reporting $\no$ infinitely often. 

Now, if $x(E) \in LIN\_O$, 
$x^\sim(E)$ might or might not be in $LIN\_O$.
If $x^\sim(E) \in LIN\_O$, then, in every iteration, every process locally computes
a finte history that is linearizable, and hence no process reports $\no$ ever.
But if $x^\sim(E) \notin LIN\_O$, every process reports $\no$ infinitely often, 
as argued in the previous paragraph,
which is incorrect to achieve strong decidability.
However, Theorem~\ref{theo:prev-work}(2) implies that there is an execution of $V$
that is indistinguishable to all processes and whose input is precisely~$x^\sim(E)$.
Thus, $x^\sim(E)$ is indeed a behavior that $\advt$ is able to produce. 
Thinking $\advt$ as a distributed service, in this case
processes somehow have used $x^\sim(E)$ to ``predict'' that $\advt$ \emph{is not} linearizable,
although its current behavior $x(E)$ is linearizable.
Observe that only an external global observer is able to determine that the
current behavior of $\advt$ is $x(E)$ and not $x^\sim(E)$, but for the 
processes it is impossible to tell this, and to them 
it is perfectly possible that the current behavior of $\advt$ is $x^\sim(E)$. 

This motivates the next decidability notion~\cite{podc23}[Definition 6.1], 
{a relaxation of} strong decidability in which processes are allowed to make
``mistakes'' when $x(E)$ is in the decided language, as long
as they have a proof that $\advt$ is not correct.

\begin{definition}[Predictive Strong Decidability]
An algorithm $V$ \emph{predictively strongly decides} a language $L$ 
against the timed adversary $\advt$, if in every execution $E$,
\begin{itemize}
\item[]
$x(E) \in L \implies \forall p, \no(E,p) = 0 \vee \big( \exists p, \no(E,p) > 0 \, \wedge \, x^\sim(E)~\notin~L \, \wedge
\big(\exists \hbox{exec. $E'$ of\, $V$ s.t. } \\E'~\equiv~E~\wedge~x(E') = x^\sim(E)\big) \big)$ 

\item[]
$x(E) \notin L \implies \exists p, \no(E,p) > 0$
\end{itemize}

Alternatively, we say that $L$ is \emph{predictively strongly decidable}.
The class of predictively strongly decidable languages is denoted $\psd$.
\end{definition}

Observe that there is no decidability predicate $\sf P$ 
(which only states properties of report values in executions)
such that $\sf P$-decidability (considering $\advt$ instead of $\adv$)
corresponds to predictive strong decidability, as when $x(E) \in L$, 
if a process reports $\no$ in $E$, it requires the existence of 
executions of $V$ satisfying specific properties.

Theorem 8.1 in~\cite{podc23} proves the correctness of algorithm $V_O$, which
implies that in general linearizability is predictively-strongly decidable.

\begin{theorem}
For any sequential object~$O$, its associated linearizability language $LIN\_O \in \psd$.
\end{theorem}

As argued in~\cite{podc23}, 
the previous result can be extended to generalizations of linearizability such as 
\emph{set linearizability}~\cite{N94} and \emph{interval linearizability}~\cite{CRR18, CRR23}, which 
are proposed to deal with correctness of implementations  of \emph{inherently} concurrent objects
that scape linearizability. The case of interval linearizability is remarkable as it has been 
shown to be a complete specification formalism, under reasonable assumptions~\cite{CRR18,GLM18}.

\paragraph{Replacing snapshots with collects.}
It is not evident that there is an alternative construction to that in~\cite{podc23}
to obtain $h_i$ in  $V_O$ (Line~\ref{VL5}) so that the aforementioned results hold 
when snapshots are replaced with weaker collects operations in $\advt$ and $V_O$.
Recently, it was shown that indeed such construction exists~\cite{rv24},
however, replacing snapshots with collects makes the construction of $h_i$
and analysis of $V_O$ more complex.

\subsection{Predictive weak decidability}

Although diminished, $\advt$ is still powerful enough to preclude some 
eventual correctness properties to be predictively strongly decidable:

\begin{lemma}
\label{lemma:no-psd-evc}
$WEC\_ COUNT, SEC\_COUNT \notin \psd$.
\end{lemma}

\begin{proof}
The proof is similar to that of Lemma~\ref{lemma:no-sd-evc}, with the difference
that views in the responses of $\advt$ need to be taken into account.

For $n=2$, by contradiction, suppose that there is an algorithm $V$ that predictively strongly decides $WEC\_COUNT$,
and consider the word $x \notin WEC\_COUNT$ where $p_1$ executes $add()$, 
and then $p_2$ and $p_1$ alternatively execute infinitely many $read()$ operations returning~0:
$$<^+_1 \hbox{ } >^{\,}_1 \hbox{    } <^{\,}_2 \hbox{ }  >^0_2   \hbox{    }  <^{\,}_1 \hbox{ }  >^0_1
 \hbox{    } <^{\,}_2 \hbox{ }  >^0_2   \hbox{    }  <^{\,}_1 \hbox{ }  >^0_1  \hbox{    } \hdots$$

Clearly $x \notin WEC\_COUNT$. We can easily adapt the proof of Claim~\ref{claim:execs}
to argue that there is an execution $E$ of $V$ such that $x(E) = x$,
where each process
atomically executes Lines~\ref{Gen1}-\ref{Gen3} of $V$ together with Lines~\ref{AT1}-\ref{AT3} of $\advt$,
and atomically executes Lines~\ref{Gen4}-\ref{Gen6} together with  Lines~\ref{AT4}-\ref{AT7} of~$\advt$.
This type of histories of $\advt$ are called \emph{tight} in~\cite{podc23}, and they have the property that 
\emph{inputs are equal to sketches}, namely, $x(E) = x^\sim(E)$. 
Let $F$ be the shortest (finite) prefix of $E$ in which a process reports $\no$.
Without loss of generality, assume that such process reporting $\no$ is $p_2$.
Thus, at the end of $F$, $p_2$ executes its block of code corresponding to Line~\ref{Gen6}, reporting $\no$
(for the first time in $F$).
Let us consider the finite input $x(F)$ in prefix $F$ (namely, the projection of symbols in $\Sigma$).
Observe that $x(F)$ is a finite prefix of $x$ that ends with $<^{\,}_2 \hbox{ }  >^0_2$.
Consider the $\omega$-word $x'$:
$$x(F) \hbox{    }  <^{\,}_1 \hbox{ }  >^1_1 \hbox{    }  <^{\,}_2 \hbox{ }  >^1_2 \hbox{    }  <^{\,}_1 \hbox{ }  >^1_1 \hbox{    }  <^{\,}_2 \hbox{ }  >^1_2 \hdots$$
Note that $x' \in WEC\_COUNT$. By the modified proof of Claim~\ref{claim:execs},
there is an execution $E'$ of $V$ such that $x(E') = x'$,
and moreover, $F$ is prefix of $E'$. Thus, $p_2$ reports $\no$ in~$E'$.
Since (1) $V$ is assumed to predictively strongly decide $WEC\_COUNT$, (2) $x' \in WEC\_COUNT$
and (3) a process reports $\no$ in~$E'$, then it must be that $x^\sim(E') \notin  WEC\_COUNT$.
Here we reach a contradiction because $E'$ is a tight execution, as defined above, and hence 
$x' = x(E') = x^\sim(E') \in WEC\_COUNT$. Thus, $V$ does not predictively strongly decides $WEC\_COUNT$.

Finally, since $SEC\_COUNT \subset WEC\_COUNT$, we also have that $SEC\_COUNT \notin \sd$.
\end{proof}

The previous impossibility result motivates the next predictive version of $\wde$,
which actually allows the existence of algorithms deciding $SEC\_COUNT$, as shown after:

\begin{definition}[Predictive Weak Decidability]
An algorithm $V$ \emph{predictively weakly decides} a language $L$ 
against the timed adversary $\advt$, if in every execution $E$,
\begin{itemize}
\item[]
$x(E) \in L \implies \forall p, \no(E,p) < \infty \vee \big( \exists p, \no(E,p) = \infty \, \wedge \, x^\sim(E)~\notin~L \, \wedge
\big(\exists \hbox{exec. $E'$ of\, $V$ s.t. } \\E'~\equiv~E~\wedge~x(E') = x^\sim(E)\big) \big)$ 

\item[]
$x(E) \notin L \implies \forall p, \no(E,p) = \infty$
\end{itemize}
Alternatively, we say that $L$ is \emph{predictively weakly decidable}.
The class of predictively weakly decidable languages is denoted $\pwd$.
\end{definition}

\begin{lemma}
$\adv \hbox{ is a strong eventual counter } \Longleftrightarrow \advt \hbox{ is a strong eventual counter (ignoring views)}$.
\end{lemma}

\begin{proof}
First, observe that in a history $x$ of $\advt$, every operation $(v,w)$ has an ``inner'' operation
$(v,w)$ (i.e. with same invocation and response) that is produced by $\adv$. 
Let $x'$ denote the projection with the history of $\adv$ in $x$.
Observe that in $x$ and $x'$ every process executes the same sequence of operations
(namely, $x|_i = x'|_i$, for every process $p_i$), but precedence and concurrence relations might be different.

To prove the $\Rightarrow$ direction, consider any history $x$ of $\advt$. 
The observations made above imply that if $x'$ satisfies the first three properties of the weak eventual counter,
then $x$ does too. As for the fourth property of the strong eventual counter, note that $x$ is basically $x'$ of $\adv$ where
operations might be ``stretched'', as the operations of $x'$ are nested in the operations o $x$. 
This implies that if an operation $(v,w)$ is concurrent to an $op$ operation in $x'$, then
the two operations are concurrent in $x$ too, and if $(v,w)$ is preceded by $op$ in $x'$,
then either that precedence relation is preserved in $x$ or the operations are concurrent. 
This observation implies that if $read$ operations satisfy the fourth property in $x'$, then they satisfy it in $x$.

To prove the $\Leftarrow$ direction, we consider the contrapositive.
Hence, suppose that there is a history $x'$ of $\adv$ that does not satisfy the strong eventual counter properties.
Observe that asynchrony allows a history of $x$ of $\adv$ in which 
(1) $\adv$ exhibits behavior $x$,
(2) in every iteration of every process, Lines~\ref{AT1}-\ref{AT3} are executed atomically,
i.e., one after the other with no step of other processes in between, and
(3) in every iteration of every process, Lines~\ref{AT4}-\ref{AT7} are executed atomically too.
Note that $x = x'$, and thus a history of $\advt$ is not consistent with the strong eventual consistent counter. 
\end{proof}

\begin{lemma}
\label{lemma:couner-in-pwd}
$SEC\_COUNT \in \pwd$.
\end{lemma}

\begin{proof}
Consider algorithm $V$ that appears in Figure~\ref{fig-algo-sec-count}.
Basically, $V$ is a modification of the algorithm in Figure~\ref{fig-algo-wec-count},
that, with the help of the views,
additionally tests if the current behavior of $\advt$ satisfies the fourth property in the definition 
of the strong eventual counter, at the end of the block in its Line~\ref{SEC6}. 
The new code appears in blue.

\begin{figure}[htb]
\centering{ \fbox{
\begin{minipage}[t]{150mm}
\scriptsize
\renewcommand{\baselinestretch}{2.5} \resetline
\begin{tabbing}
aaaa\=aaa\=aaa\=aaa\=aaa\=aaa\=aaa\=\kill 

{\bf Shared variables in $V$:}\\

$~~$  $INCS[1,\hdots,n]:$ shared array of read/write registers, each initialized to $0$\\

$~~$  $M[1, \hdots, n]:$ shared array of read/write registers, each initialized to $\emptyset$\\ \\

{\bf Local algorithm $V_i$ for process $p_i$:}\\

 \> $prev\_read_i \leftarrow 0$\\
 
 \> $prev\_incs_i \leftarrow 0$\\
 
  \> $count_i \leftarrow 0$\\
 
 \> $flag_i \leftarrow \false$\\

 \> $s_i = \emptyset$\\

 \> {\bf while} {\sf true} {\bf do}\\

\line{SEC1} \>\> Non-deterministically pick an invocation symbol $v_i \in \Sigma^<_i$\\

\line{SEC2} \>\> {\bf if} $v_i$ is an invocation to $inc()$ {\bf then} \\

 \>\>\> $count \leftarrow count + 1$\\
 
  \>\>\> $INCS[i].write(count)$ \\

\line{SEC3} \>\> Send $v_i$ to the adversary $\advt$\\

\line{SEC4} \>\> Receive $(w_i, view_i)$, $w_i \in \Sigma^>_i$, from the adversary $\advt$\\

\line{SEC5} \>\> $snap_i \leftarrow \snap(INCS)$\\

\>\> $curr\_incs_i \leftarrow snap_i[1] + \hdots + snap_i[n]$\\

\>\> {\bf if} $w_i$ is a response to $read()$ {\bf then} $curr\_read_i \leftarrow$ returned value in $w_i$\\

\>\> \color{blue} $s_i \leftarrow s_i \cup \{ (v_i, w_i, view_i) \} $\\

 \>\> \color{blue} $M[i].write(s_i)$\\

\>\> \color{blue} $snap'_i \leftarrow \snap(M)$\\

\line{SEC6} \>\> {\bf if} $flag_i == \true$ {\bf then}\\ 

\>\>\> {\bf report} $\no$\\

\>\> {\bf elseif} $curr\_read_i < snap_i[i] \, \vee \, curr\_read_i < prev\_read_i$ {\bf then}\\

\>\>\> $flag_i \leftarrow \true$\\

\>\>\> {\bf report} $\no$\\

\>\> {\bf elseif} $curr\_read_i \, != curr\_incs_i \, \vee \, prev\_incs_i < curr\_incs_i$ {\bf then}\\

\>\>\> {\bf report} $\no$\\

\>\> \color{blue} {\bf elseif} $\exists (v_j,w_j,view_j)$ in $snap'_i$ such that $v_j$ is an invocation to $read()$ and\\ 
\>\> \> \> \color{blue} the returned value in $w_j$ is larger than the number of invocations to $inc()$ in $view_j$ {\bf then}\\

\>\>\> \color{blue} {\bf report} $\no$\\

\>\> {\bf else}\\ 

\>\>\> {\bf report} $\yes$\\

 \>\> $prev\_read_i \leftarrow curr\_read_i$\\
 
 \>\> $prev\_incs_i \leftarrow curr\_incs_i$   

\end{tabbing}
\end{minipage}
  }
\caption{Predictively weakly deciding $SEC\_COUNT$.}
\label{fig-algo-sec-count}
}
\end{figure}

Let $E$ be any execution of $V$. Consider first the case $x(E) \notin SEC\_COUNT$.
As the proof of Lemma~\ref{fig-algo-wec-count} shows, if $x(E)$ does not satisfy one
of the first three properties that define the strong eventual counter, then
every process reports $\no$ infinitely often in $E$.
If $x(E)$ satisfies those properties but fails to satisfy the fourth one,
then again we can argue that all process reports $\no$ infinitely often in $E$.
The reason is that $x^\sim(E)$ does not satisfy the fourth property either
(namely, $x^\sim(E) \notin SEC\_COUNT$),
due to the fact that $x^\sim(E)$ is $x(E)$ with some operations ``stretched'' (see Appendix~\ref{app:construction}),
and thus $x^\sim(E)$ preserves the precedence relations in $x(E)$, by Theorem~\ref{theo:prev-work}(1),
but some operations that are concurrent in $x(E)$ are not concurrent in $x^\sim(E)$.
Since $E$ is fair, eventually a $read$ operation failing to satisfy the fourth property
is written in $M$, and thus eventually every process infinitely often reports $\no$,
due to the fourth condition in Line~\ref{SEC6}.

Now, consider the case $x(E) \in SEC\_COUNT$.
We have two sub-cases, depending whether $x^\sim(E)$ is in $SEC\_COUNT$ or not.
If $x^\sim(E) \notin SEC\_COUNT$, then, as argued above, every process reports $\no$ infinitely often,
and Theorem~\ref{theo:prev-work}(2) shows that there is an execution $E'$ of $V$ such that
$E \equiv E'$ and $x(E') = x^\sim(E)$.
In the second sub-case, $x^\sim(E) \in SEC\_COUNT$, the proof of Lemma~\ref{fig-algo-wec-count} shows
that the first three conditions in Line~\ref{SEC6} hold only finitely any times.
As for the fourth condition in the same line, it never holds in the execution as 
$x^\sim(E) \in SEC\_COUNT$. Therefore, every process reports $\no$ only finitely many times. 

We conclude that $V$ predictively-weakly decides $SEC\_COUNT$, and thus $SEC\_COUNT~\in~\pwd$.
\end{proof}

Therefore, we have the following separation result, where the inclusion of $\psd$ in $\pwd$ directly follows
from the definitions.

\begin{theorem}
\label{theo:psd-in-pwd}
$\psd \subset \pwd$.
\end{theorem}

Finally, we show that some languages remain undecidable 
even if we consider  predictive weak decidability.
Intuitively, the next proof shows that if there is an algorithm that predictively-weakly decides $EC\_LED$,
then one can inductively construct an execution $E$ of the algorithm such that $x(E) \in EC\_LED$
and every process reports infinitely many times $\no$ in $E$ with $x(E) = x^\sim(E)$, which
is a contradiction as the relaxation in the definition of predictive weak decidability can only happen if $x^\sim(E) \notin EC\_LED$.

\begin{lemma}
\label{lemma:counter-no-pwd}
$EC\_LED \notin \pwd$.
\end{lemma}

\begin{proof}
For $n=2$, by contradiction, suppose that there is an algorithm $V$ that predictively weakly decides $EC\_LED$.
The proof consists in inductively constructing an execution $F^\infty$ of $V$ such that 
$x(F^\infty) \in EC\_LED$, $x(F^\infty) = x^\sim(F^\infty)$ and the two processes reports infinitely many times $\no$ in~$F^\infty$,
which contradicts the existence of $V$.

Consider the word $x$ where $p_1$ executes $append(a)$, with $a \in U$,
and then $p_2$ and $p_1$ alternatively execute infinitely many $get()$ operations returning~$\epsilon$:
$$<^a_1 \hbox{ } >^{\,}_1 \hbox{    }  <^{\,}_2 \hbox{ }  >^\epsilon_2 \hbox{    } <^{\,}_1 \hbox{ }  >^\epsilon_1 \hbox{    }  
<^{\,}_2 \hbox{ }  >^\epsilon_2 \hbox{    }  <^{\,}_1 \hbox{ }  >^\epsilon_1 \hbox{    }  <^{\,}_2 \hbox{ }  >^\epsilon_2 \hbox{    } \hdots$$

Clearly, $x$ does not satisfy the second property in the definition of $EC\_LED$ (the first property is satisfied, 
however, as for every finte prefix one can place $<^a_1 \hbox{} >^{\,}_1 $
at the end to obtain a valid sequential ledger history).
Hence, $x \notin EC\_LED$. 

As explained in the proof of Lemma~\ref{lemma:no-psd-evc},
the proof of Claim~\ref{claim:execs} can be easily modified
to argue that there is an execution $E$ of $V$ such that $x(E) = x$,
where each process
atomically executes Lines~\ref{Gen1}-\ref{Gen3} of $V$ together with Lines~\ref{AT1}-\ref{AT3} of $\advt$,
and atomically executes Lines~\ref{Gen4}-\ref{Gen6} together with  Lines~\ref{AT4}-\ref{AT7} of~$\advt$.
As mentioned there, this type of executions have the property that $x(E) = x^\sim(E)$. 

Let $E'$ be any finite prefix of $E$ that ends with $p_2$ executing 
its block of code corresponding to Line~\ref{Gen6}, and each process reports $\no$ at least once in $E'$.
Such prefix exists as $V$ is assumed correct and $x \notin EC\_LED$.
Let us consider the finite input $x(E')$ in prefix $E'$ (namely, the projection of symbols in $\Sigma$).
Note that $x(E')$ is a finite prefix of $x$ that ends with $<^{\,}_2 \hbox{}  >^\epsilon_2$.
Consider the $\omega$-word $x^1$:
$$x(E') \hbox{    }  <^{\,}_1 \hbox{ }  >^{a}_1 \hbox{    }  
<^{\,}_2 \hbox{ }  >^{a}_2 \hbox{    }  <^{\,}_1 \hbox{ }  >^{a}_1 \hbox{    }  <^{\,}_2 \hbox{ }  >^{a}_2 \hbox{    } \hdots$$

Observe that $x^1 \in WEC\_COUNT$. By the modified proof of Claim~\ref{claim:execs},
there is an execution $F^1$ of $V$ such that $x(F^1) = x^1$,
and moreover, $E'$ is prefix of $F^1$ and $x(F^1) = x^\sim(F^1)$.
Thus, we have that $x(F^1) \in EC\_LED$, $x(F^1) = x^\sim(F^1)$ and 
$F^1$ has a finite prefix, $E'$, where each process reports $\no$ at least one time.
Our next goal is to argue that we can modify $F^1$ to apply the previous step
to obtain a similar execution $F^2$ with a finite prefix where each process reports $\no$ \emph{at least two times},
which shows that the construction can be taken to the infinity to obtain an execution $F^\infty$ with the desired properties.

Now, since $V$ is assumed correct and $x(F^1) = x^\sim(F^1)$, it must be that every process 
reports $\no$ finitely many times  in $F^1$.
Let $F'$ be any finite prefix of $F^1$ that has $E'$ as a prefix,
ends with $p_2$ executing 
its block of code corresponding to Line~\ref{Gen6}, and no process reports $\no$ in the suffix after $F'$.
Hence, every process reports $\no$ at least once in $F'$.
Note that $x(F')$ ends with  $<^{\,}_2 \hbox{}  >^a_2$.
For $b \in U$ that does not appear in an invocation to $append$ in $x(F')$,
consider the $\omega$-word $x'$:
$$x(H) <^b_1 \hbox{ } >^{\,}_1 \hbox{    }  <^{\,}_2 \hbox{ }  >^a_2 \hbox{    } <^{\,}_1 \hbox{ }  >^a_1 \hbox{    }  
<^{\,}_2 \hbox{ }  >^a_2 \hbox{    }  <^{\,}_1 \hbox{ }  >^a_1 \hbox{    }  <^{\,}_2 \hbox{ }  >^a_2 \hbox{    } \hdots$$

Note that $x' \notin EC\_LED$. As above, there is an execution $H$ of $V$ such that $x(H) = x'$,
$F'$ is prefix of $H$ and $x(H) = x^\sim(H)$.
We have that every process reports $\no$  infinitely many times in $H$.
We can now repeat the first step of the construction, taking $E'$ as
any finite prefix of $H$ that ends with $p_2$ executing 
its block of code corresponding to Line~\ref{Gen6}, and each process reports $\no$ \emph{at least two times} in $E'$.
The lemma follows. 
\end{proof}

\section{Final Remarks}

For concreteness, we focused on decidability definitions using only $\yes/\no$ values.
However, it is useful to consider definitions with more than two values. 
For example, we can define a 3-value variant of $\wde$ where processes can report $\yes$, $\no$ and $\maybe$,
requiring that if the current behavior of $\adv$ is in the language, then no process reports $\no$ ever,
and otherwise no process reports $\yes$ ever. Thus, a process is allowed to report $\maybe$ 
if currently it does not have conclusive information about the current behavior of $\adv$,
and if a process ever reports $\yes$/$\no$, it is sure that the current behavior is/is not correct. 
This 3-value decidability definition is reminiscent to the 3-value LTL 
that has been used in the past in centralized runtime verification~\cite{BLS11}.
It is easy to adapt the algorithm in Figure~\ref{fig-algo-wec-count} to argue that 
$WEC\_COUNT$ is 3-value $\wde$ decidable (it suffices to change $\yes$ with $\maybe$ in the last block of code). 
A similar 3-value variant of $\pwd$ can be defined, and shown that $SEC\_COUNT$ belongs to this class.
There have been proposed monitors
that allows several report values, with lower bounds shown on the number of report values needed to 
runtime verify some properties (e.g.,~\cite{BLS11,BLS10,BFRR22,FRT20}). 
This is a line of research worth exploring in our setting.

Other questions that we believe are interesting to study are the following.
We conjecture that only \emph{trivial} languages belong to $\sd$. Triviality means
that the languages define distributed problems that can be implemented with no communication
among processes.
We conjecture that the complement of $EC\_LED$ is in $\pwd$. A related more general question is
if there are language that neither them nor their complement belong to $\pwd$.
The timed adversary $\advt$ was obtained using only read/write registers. 
However, it could be the case that more powerful primitives allow the existence of \emph{strictly weaker} 
timed adversaries that make some languages $\psd$- or $\pwd$-decidable (for example, $EC\_LED$).
It is also interesting to explore if timing assumption allows to verify real-time sensitive properties
against $\adv$.

\appendix

\section{Three non real-time oblivious languages}
\label{app:languages}

Consider the following history:
$$x = <^1_1 \hbox{  } >^{}_1 \, <^2_2 \hbox{  } >^{}_2 \,  \hdots \, <^n_n \hbox{  } >^{}_n \, <^{}_n \hbox{  } >^{1 \cdot 2 \cdots n}_n  
\hbox{ } \hdots \hbox{ } <^1_1 \hbox{  } >^{}_1 \, <^2_2 \hbox{  } >^{}_2 \,  \hdots \, <^n_n \hbox{  } >^{}_n \, <^{}_n \hbox{  } >^{1 \cdot 2 \cdots n \cdot 1 \cdot 2 \cdots n}_n\hdots$$

It is not difficult to verify that $x$ is in $LIN\_LED$, $SC\_LED$ and $EC\_LED$.
Consider the prefix 
$\alpha = <^1_1 \hbox{  } >^{}_1 \, <^2_2 \hbox{  } >^{}_2 \,  \hdots \, <^n_n \hbox{  } >^{}_n \, <^{}_n \hbox{  } >^{1 \cdot 2 \cdots n}_n$
of $x$. Observe that the shuffle 
$\alpha' = <^2_2 \hbox{  } >^{}_2 \,  \hdots \, <^n_n \hbox{  } >^{}_n \, <^{}_n \hbox{  } >^{1 \cdot 2 \cdots n}_n  <^1_1 \hbox{  } >^{}_1$
is not linearizable. Also, the prefix $<^2_2 \hbox{  } >^{}_2 \,  \hdots \, <^n_n \hbox{  } >^{}_n \, <^{}_n \hbox{  } >^{1 \cdot 2 \cdots n}_n$ of $\alpha'$
is not sequentially content, and does not follow the first property in the definition of $EC\_LED$.
Hence the history
$x' = \alpha' \, <^1_1 \hbox{  } >^{}_1 \, <^2_2 \hbox{  } >^{}_2 \,  \hdots \, <^n_n \hbox{  } >^{}_n \, <^{}_n \hbox{  } >^{1 \cdot 2 \cdots n \cdot 1 \cdot 2 \cdots n}_n\hdots$
is not in $LIN\_LED$, $SC\_LED$ and $EC\_LED$, which shows that the languages are not real-time oblivious.

\section{From views to histories}
\label{app:construction}

This section explains the construction in~\cite{podc23}[Section 7] that, given an execution $E$ of
an algorithm $V$ interacting with $\advt$, produces a concurrent history $x^\sim(E)$ from the views of operations in $x(E)$.

The construction is based on a property of views directly implied by the snapshot operation:
the views of any two operations are comparable, namely, either they are equal or 
one of them strictly contains the other. 
Recall that it is assumed that each invocation symbol appears at most once in $E$.
The construction is as follows. 
All \emph{distinct} views in $x(E)$ are ordered 
in ascending containment order: $views_1 \subset view_2 \subset \hdots \subset view_\ell \subset view_{\ell+1} \subset \hdots$.
Let $\sigma_0$ denote $\emptyset$.
For each $k = 1, 2, \hdots, $ (in ascending order), $x^\sim(E)$ is iteratively obtained following the
next two steps in order:
\begin{enumerate}
\item For each invocation symbol $v \in view_k \setminus view_{k-1}$,
the invocation $v$ is appended to $x^\sim(E)$; the invocations are appended in any arbitrary order.

\item For each operation $(v,w)$ of $x(E)$ whose view is $view_k$ (i.e. $\advt$ replies $(w,view_k)$ to invocation $v$),
the response $w$ is appended to $x^\sim(E)$; the responses are appended in any arbitrary order.
\end{enumerate}

In the two steps of the construction, 
either a set of invocations or responses are placed in some arbitrary sequential order.
For any of these orders, the resulting history $x^\sim(E)$ has the same 
precedence operation relations. 
Thus, in fact, $x^\sim(E)$ denotes an equivalence class of histories.

In Figure~\ref{fig:execution}, the first iteration of the construction appends to $x^\sim(E)$ first invocations $\{$ and $[$, 
and then responses $]$ and $\}$, as operations $\{ \, \, \}$ and $[ \, \, ]$ have view $\{ \, [$; 
the second iteration appends invocation $\lceil$ and then response~$\rceil$,
as operation $\lceil \, \, \rceil$ has view $\{ \, [ \, \lceil$;
and the third iteration appends invocation $\langle$ and then response $\rangle$.

By construction, the operations that precede or are concurrent to an operation $op$ in $x^\sim(E)$,
are those whose invocations appear in the view of $op$ in $E$.

\begin{acks}
  Armando Castañeda is supported by the research project DGAPA-PAPIIT IN108723.
  Gilde Valeria Rodríguez is the recipient of a PhD fellowship of SECIHTI.
\end{acks}

\bibliographystyle{abbrv}
\bibliography{biblio}

\end{document}